\documentclass{lmcs}
\usepackage[utf8]{inputenc}
\usepackage[T1]{fontenc}

\usepackage{csquotes}

\usepackage{bm}

\usepackage{caption}
\usepackage{subcaption}

\usepackage{multirow}

\usepackage{booktabs}

\usepackage{fancybox}
\usepackage{xparse}
\usepackage{etoolbox}

\usepackage{graphicx}

\usepackage{fontawesome}

\usepackage{hyperref}
\usepackage{cleveref}

\AtBeginEnvironment{exa}{\pushQED{\qed}}
\AtEndEnvironment{exa}{\popQED}
\AtBeginEnvironment{exas}{\pushQED{\qed}}
\AtEndEnvironment{exas}{\popQED}
\AtBeginEnvironment{rem}{\pushQED{\qed}}
\AtEndEnvironment{rem}{\popQED}
\AtBeginEnvironment{defi}{\pushQED{\qed}}
\AtEndEnvironment{defi}{\popQED}
\AtBeginEnvironment{prob}{\pushQED{\qed}}
\AtEndEnvironment{prob}{\popQED}

\crefname{exa}{Example}{Examples}
\crefname{exas}{Examples}{Examples}
\crefname{prob}{Problem}{Problems}
\crefname{thm}{Theorem}{Theorems}
\crefname{prop}{Proposition}{Proposition}
\crefname{thmC}{Theorem}{Theorems}
\crefname{lem}{Lemma}{Lemmas}
\crefname{cor}{Corollary}{Corollaries}
\crefname{defi}{Definition}{Definitions}
\crefname{rem}{Remark}{Remarks}

\let\oldphi\phi
\let\phi\varphi
\let\varphi\oldphi

\let\epsilon\varepsilon

\let\from\colon
\let\with\colon
\newcommand*{\after}{\mathbin{\circ}}

\let\vec\bm

\newcommand*{\T}{\mathsf{true}}
\newcommand*{\F}{\mathsf{false}}
\newcommand*{\RR}{\mathbb{R}}
\newcommand*{\powerset}[1]{2^{#1}}
\DeclareMathOperator*{\E}{\mathbb{E}}

\DeclareMathOperator{\Shap}{\text{\normalfont\textsc{Shap}}}
\DeclareMathOperator{\Shapley}{\ensuremath{\mathrm{Shapley}}}

\DeclarePairedDelimiter{\abs}{\lvert}{\rvert}
\DeclarePairedDelimiter{\card}{\lvert}{\rvert}
\DeclarePairedDelimiter{\set}{\lbrace}{\rbrace}
\DeclarePairedDelimiter{\enc}{\lVert}{\rVert}
\DeclarePairedDelimiter{\len}{\lvert}{\rvert}

\DeclareMathOperator{\supp}{\mathop{\mathsf{supp}}}

\DeclareMathOperator{\parsupp}{\mathop{\mathsf{psupp}}}

\DeclareMathOperator{\dom}{\mathop{\mathsf{dom}}}
\DeclareMathOperator{\inschema}{\mathop{\mathsf{In}}}
\DeclareMathOperator{\outschema}{\mathop{\mathsf{Out}}}
\DeclareMathOperator{\parschema}{\mathop{\mathsf{Par}}}
\DeclareMathOperator{\adom}{\mathop{\mathsf{adom}}}

\DeclareMathOperator{\DB}{DB}
\DeclareMathOperator{\Rel}{Rel}
\DeclareMathOperator{\Tup}{Tup}
\NewDocumentCommand{\paramquery}{mom}{
    \IfNoValueTF{#2}
        {#1_{#3}}
        {#1(#2;#3)}
}

\newcommand{\pquery}[3]{#1(#2;#3)}
\newcommand{\pqueryinst}[2]{#1_{#2}}

\newcommand*{\complexityfont}[1]{\mathsf{#1}}
\newcommand*{\sharpP}{\ensuremath{\complexityfont{\#P}}}

\newcommand*{\cookleq}{\mathrel{\leq^{\complexityfont{P}}_{\complexityfont{T}}}}
\newcommand*{\cookequiv}{\mathrel{\equiv^{\complexityfont{P}}_{\complexityfont{T}}}}

\newcommand*{\problemname}[1]{\operatorname{\mathsf{#1}}}
\newcommand*{\NONEMPTY}{\problemname{NONEMPTY}}
\newcommand*{\SHAP}{\problemname{SHAP}}
\newcommand*{\ESIM}{\problemname{ESIM}}
\newcommand*{\sharpPDNF}{\problemname{\#posDNF}}
\newcommand*{\sharpwACQ}{\problemname{\#_wACQ}}

\newcommand*{\WhyNotShapleyQual}{\ensuremath\problemname{WhyNotShapley}_{0\text{-}1}}
\newcommand*{\WhyNotShapleyQuan}{\ensuremath\problemname{WhyNotShapley}_{\mathrm{size}}}
\newcommand*{\WhyNotSHAPQuan}{\ensuremath\problemname{WhyNotSHAP}_{\mathrm{size}}}
\newcommand*{\sharpSetCover}{\ensuremath{\problemname{\#SetCover}}}
\newcommand*{\sharpClique}{\ensuremath{\problemname{\#Clique}}}
\newcommand*{\probShapley}{\ensuremath{\problemname{Shapley}}}
\newcommand*{\binSHAP}{\ensuremath{\problemname{binarySHAP}}}

\newcommand*{\distclass}[1]{\mathtt{#1}}
\newcommand*{\Q}{\mathcal Q}

\newcommand*{\PR}{\mathcal{P}}
\newcommand*{\IND}{\mathcal{P}^{\distclass{IND}}}

\newcommand*{\Qclique}{\mathcal Q_{\mathrm{clique}}}

\newcommand*{\simi}{\mathfrak{s}}
\newcommand*{\simifont}[1]{\mathsf{#1}}
\DeclareMathOperator{\Jaccard}      {\simifont{Jaccard}}
\DeclareMathOperator{\Count}        {\simifont{Count}}

\DeclareMathOperator{\Intersection} {\simifont{Int}}
\DeclareMathOperator{\NegSymDiff}   {\simifont{NegSymDiff}}
\DeclareMathOperator{\NegSymCDiff}     {\simifont{NegSymCDiff}}
\DeclareMathOperator{\NegDiff}      {\simifont{NegDiff}}

\newcommand*{\pACQ}{\ensuremath{\mathrm{pACQ}}}

\newsavebox{\fproblembox}

\newcommand{\vqual}{\nu^{\mathrm{0\text{-}1}}}
\newcommand{\vquan}{\nu^{\mathrm{size}}}
\newcommand{\emptytuple}{\bm{\varepsilon}}

\def\e#1{\emph{#1}}

\crefname{exa}{Example}{Examples}
\Crefname{thm}{Theorem}{Theorems}

\usepackage{tikz}
\usetikzlibrary{positioning}
\usetikzlibrary{shapes} 
\title{The Importance of Parameters in Database Queries{\rsuper*}}
\titlecomment{{\lsuper*}A conference version of the manuscript appeared in the 2024 International Conference on Database Theory (ICDT'24)~\cite{GroheK0S24}.}

\author[A.~Gilad]{Amir Gilad\lmcsorcid{0000-0002-3764-1958}}[a]
\author[M.~Grohe]{Martin Grohe\lmcsorcid{0000-0002-0292-9142}}[b]
\author[B.~Kimelfeld]{Benny Kimelfeld\lmcsorcid{0000-0002-7156-1572}}[c]
\author[P.~Lindner]{Peter Lindner\lmcsorcid{0000-0003-2041-7201}}[d]
\author[C.~Standke]{Christoph Standke\lmcsorcid{0000-0002-3034-730X}}[b]

\address{The Hebrew University, Jerusalem, Israel}
\email{amirg@cs.huji.ac.il}

\address{RWTH Aachen University, Aachen, Germany}
\email{grohe@informatik.rwth-aachen.de, standke@informatik.rwth-aachen.de}

\address{Technion -- Israel Institute of Technology, Haifa, Israel}
\email{bennyk@cs.technion.ac.il}

\address{École Polytechnique Fédérale de Lausanne, Lausanne, Switzerland}
\email{peter.lindner@epfl.ch}

\begin{document}

\begin{abstract}
We propose and study a framework for quantifying the importance of the choices of parameter values to the result of a query over a database. 
These parameters occur as constants in logical queries, such as conjunctive queries. 
In our framework, the importance of a parameter is its $\Shap$ score -- a popular instantiation of the game-theoretic Shapley value to measure the importance of feature values in machine learning models. We make the case for the rationale of using this score by explaining the intuition behind $\Shap$, and by showing that we arrive at this score in two different, apparently opposing, approaches to quantifying the contribution of a parameter. The application of $\Shap$ requires two components in addition to the query and the database: (a) a probability distribution over the combinations of parameter values, and (b) a utility function that measures the similarity between the result for the original parameters and the result for hypothetical parameters. The main question addressed in the paper is the complexity of calculating the $\Shap$ score for different distributions and similarity measures.
In particular, we devise polynomial-time algorithms for the case of full acyclic conjunctive queries for certain (natural) similarity functions. 
We extend our results to conjunctive queries with parameterized filters (e.g., inequalities between variables and parameters). We also illustrate the application of our results to ``why-not'' explanations (aiming to explain the absence of a query answer), where we consider the task of quantifying the contribution of query components to the elimination of a non-answer in consideration. Finally, we discuss a simple approximation technique for the case of correlated parameters.
\end{abstract}

\maketitle

\setcounter{page}{1}

\section{Introduction}

The parameters of a database query may affect the result in a way that misrepresents the importance of the parameters, or the arbitrariness in their chosen values. For example, when searching for products in commercial applications (for clothing, travel, real estate, etc.), we may fill out a complex form of parameters that produce too few answers or overly expensive ones; what is the responsibility of our input values to this deficient outcome?  We may phrase a database query to select candidates for awards or job interviews; to what extent is the choice of parameters affecting people's fate? 

Considerable effort has been invested in exploring the impact of parameters on query outcomes. In the \e{empty-answer problem}, the goal is typically to explore a space of small changes to the query that would yield a nonempty result~\cite{DBLP:conf/vldb/KoudasLTV06, DBLP:journals/pvldb/MottinMRDPV13}. In that vein, reasoning about small parameter changes, or \e{perturbations}, has been applied to providing explanations to \e{non-answers} (also called \e{why-not questions}), that is, tuples that are missing from the result~\cite{ChapmanJ09,TranC10}. From a different angle, the analysis of sensitivity to parameters has been applied to \e{fact checking}, and particularly, the detection of statements that are \e{cherry picked} in the sense that they lead to conclusions that overly rely on allegedly arbitrary parameter values.
This may come in over-restriction to a database fragment that serves the intended claim~\cite{DBLP:journals/tods/0001ALYY17}, or over-generalization that masks the situation in substantial subgroups that oppose the claim~\cite{DBLP:journals/pvldb/LinYMJM21}.

In this work, we aim to establish a principled quantitative measure for the importance of individual parameter values to the result $Q(D)$ of a query $Q$ over a database $D$. 
To this end, we begin with the basic idea of observing how the result changes when we randomly change the parameter of interest. Alternatively, we can observe the change in the result when the parameter keeps its value while all others randomly change. Yet, these definitions ignore dependencies among parameters; changing a parameter may have no impact in the presence of other parameter values (e.g., the number of connecting flights does not matter if we restrict the travel duration), or it may lead to an overestimation of the value's importance (e.g., changing the number of semesters empties the result since we restrict the admission year). This can be viewed as a special case of a challenge that has been studied for decades in game theory: \emph{How to attribute individual contributions to the gains of a team?}
Specifically, we can view the parameter values as players of a cooperative game where each coalition (set of parameter values in our case) has a utility, and we wish to quantify the contribution of each parameter value to the overall utility. 
We then adopt a conventional formula for contribution attribution, namely the Shapley value~\cite{Shapley}, as done in many domains, including economics, law, bioinformatics, crime analysis, network analysis, machine learning, and more (see, e.g., the \e{Handbook of the Shapley value}~\cite{algaba2019handbook}). The Shapley value is theoretically justified in the sense that it is unique under several elementary axioms of rationality for profit sharing~\cite{Shapley}. 
In the context of databases, this value has been studied recently for measuring the contribution of individual tuples to query answers~\cite{DBLP:conf/icdt/LivshitsBKS20,DBLP:conf/sigmod/DeutchFKM22,DBLP:conf/cikm/AradDF22,DBLP:journals/pacmmod/BienvenuFL24} and to database inconsistency~\cite{DBLP:journals/lmcs/LivshitsK22}, as well as
the contribution of constraints to the decisions of cleaning systems~\cite{DBLP:conf/cikm/DeutchFGS21}. 

Our challenge then boils down to one central question: What game are we playing? In other words, what is the utility of a set $J$ of parameter values? Following up on the two basic ideas discussed above, we can think of two analogous ways. In the first way, we measure the change in the query result when we randomly change the values of the parameters in $J$; the parameter values in $J$ are deemed important if we observe a large change. This change is random, so we take the \emph{expected} change. (We later discuss the way that we measure the change in the result.) In the second way, we again measure the change in the result, but now we do so when we fix the values of the parameters in $J$ and randomly change the rest; now, however, the values in $J$ are deemed important if we observe a \emph{small} change, indicating that the other parameters have little impact once we use the values of $J$. 
This second way is known as the $\Shap$ score~\cite{lund17,lundberg20} in machine learning,
and it is one of the prominent score-attribution methods for features (in addition to alternatives such as LIME~\cite{DBLP:conf/kdd/Ribeiro0G16} and
Anchor~\cite{anchors:aaai18}).\footnote{For background on attribution scores see textbooks on explanations in machine learning, e.g., Molnar~\cite{molnar2022}.} This score
 quantifies the impact of each feature value on the decision for a specific given instance.  Interestingly, we show that the first way described above also coincides with the $\Shap$ score, so the two ways actually define the \emph{same measure} (\Cref{thm:same_measure}). We prove it in a general setting and, hence, this equivalence is of independent interest as it shows an alternative, apparently different way of arriving at $\Shap$ in machine learning.

\subsection*{Basic setup.}
 To materialize the framework in the context of a query $Q$ and a database $D$, one needs to provide some necessary mechanisms for reasoning about $Q(D)$ and $Q'(D)$, where $Q'$ is the same as $Q$ up to the parameters: it uses the same values as $Q$ for the parameters of $J$, but the remaining parameter values are selected randomly. Specifically, to this aim, we need two mechanisms:
\begin{enumerate}
    \item A \emph{probability distribution} $\Gamma$ of possible parameterizations of the query;
    \item A \emph{similarity function} $\simi$ between relations to quantify how close $Q'(D)$ is to $Q(D)$.
\end{enumerate}
The distribution $\Gamma$ may include any feasible combination of parameter values, and they can be either probabilistically independent or correlated. For $\simi$, one can use any similarity between sets (see, e.g., surveys on similarity measures such as~\cite{lesot:hal-01072737,DBLP:journals/jifs/SathiyaG19}) or measures that account for the distance between attributes values (e.g., as done in the context of database repairs~\cite{DBLP:journals/is/BertossiBFL08}). We give examples in \Cref{sec:shap}.

A central challenge in the framework is the computational complexity, since the direct definition of the $\Shap$ score (like the general Shapley value) involves summation over an exponential space of coalitions. Indeed, the calculation can be a hard computational problem, $\sharpP$-hard to be precise, even for simple adaptations of the $\Shap$ score~\cite{DBLP:conf/aaai/ArenasBBM21} and the Shapley value~\cite{10.1007/BF01258278,DBLP:journals/mor/DengP94,DBLP:conf/icdt/LivshitsBKS20}. Hence, instantiations of our framework require specialized complexity analyses and nontrivial algorithms that bypass the exponential time of the na\"ive computation.

We begin the complexity analysis of the framework by establishing some general insights for finite fully factorized distributions, i.e., where the parameters are probabilistically independent and each is given as an explicit collection of value-probability pairs. First, the $\Shap$ score can be computed in polynomial time if we can evaluate the query, compute the similarity measure, and enumerate all parameter combinations in polynomial time. We prove this using a recent general result by Van den Broeck et al.~\cite{van_den_broeck2022tractability} showing that, under tractability assumptions, the $\Shap$ score is reducible to the computation of the expected value under random parameter values. Second, under reasonable assumptions, the computation of the $\Shap$ score is at least as hard as testing for the emptiness of the query, for \emph{every} nontrivial similarity function; this is expected, as the definition of the $\Shap$ scores requires, conceptually, many applications of the query.

Next, we focus on the class of conjunctive queries, where the parameters are constants in query atoms. Put differently, we consider Select-Project-Join queries where each selection predicate has the form $x=p$ where $x$ is an attribute and $p$ is a parameter. 
It follows from the above general results that this case is tractable under data complexity. Hence, we focus on combined complexity. As the emptiness problem is intractable, we consider the tractable fragment of acyclic queries, and show that $\Shap$ scores can be $\sharpP$-hard even there. 

We then focus on the class of \e{full} acyclic queries (with parameters treated as variables) and establish that the $\Shap$ score can be computed in polynomial time for three natural, set-based similarity functions between $Q(D)$ and $Q'(D)$. Interestingly, this gives us nontrivial cases where the $\Shap$ score can be computed in polynomial time even if $Q(D)$ and $Q'(D)$ can be exponential in the size of the input, and hence, it is intractable to materialize them. 

\subsection*{Extensions}
Next, we extend our results to conjunctive queries with \e{filters} (\Cref{sec:filters}). A filter in a conjunctive query can be considered a conjunct that can be any Boolean condition on the assignments of an ordinary (filter-free) conjunctive query, such as a Boolean combination of built-in relations (e.g., inequalities). Filters may include parameters  (e.g., the inequality $x\geq p$ where $x$ is a variable and $p$ is a parameter) and, as usual, we are interested in their $\Shap$ score. We show that this addition can make the $\Shap$ score intractable to compute, even if there is a single relational atom in addition to simple inequality filters.
Nevertheless, we identify cases where tractability properties are retained when adding filters to classes of parameterized queries (e.g., full acyclic conjunctive queries), relying on structural assumptions on the variables and parameters involved in the filters. 
We show an application of this extension to measure the importance of query operators in ``why not'' questions (also called \emph{provenance for non-answers}) in query answering~\cite{ChapmanJ09,BidoitHT15,HuangCDN08,HerschelH10,Herschel15,MeliouGMS11,TranC10,HeL14,LiuGCZZ16}. 
Such questions arise
in scenarios where a tuple is expected to appear in the result set of a query, but is nevertheless absent from it. 
To explain this absence, previous work has considered several models of explanations.
Our work aligns with the operator-based approach~\cite{ChapmanJ09,BidoitHT14,BidoitHT15} that focuses on discovering operators that disqualify the tuple. 
Our contribution is a study of the Shapley value of the filters in the game of eliminating the answer (\Cref{sec:why-not}). We examine two natural games: computing Shapley values is $\sharpP$-hard for a qualitative game, which asks whether the filters eliminate the answer, even for queries with a single relational atom, but tractable for acyclic conjunctive queries with (efficiently evaluable) filters for a quantitative game, which counts the intermediate tuples that the filters eliminate. Interestingly, we obtain the latter result by a reduction to the computation of $\Shap$ scores of parameters.

Given that the computation of the exact $\Shap$ score is often intractable, we also study the complexity of \emph{approximate} evaluation (\Cref{sec:cor-approx}). We show that using sampling, we can obtain an efficient approximation scheme (FPRAS) with additive guarantees. Moreover, the tractability of approximation generalizes to allow for parameters that are correlated through Bayesian networks (and actually any distribution) that provide(s) polynomial-time sampling while conditioning on assignments to arbitrary subsets of the random variables.

\subsection*{Comparison to a conference version.}
An abridged version of this work appeared in a conference version~\cite{GroheK0S24}. In comparison to that version, this manuscript includes several  new additions. First, the manuscript includes the full proofs of all of the results. Second, the study of filters (\Cref{sec:filters}) is new. (The conference version included a restricted discussion on the addition of inequalities to parameterized conjunctive queries.) Third, also entirely new is the application to why-not questions, including the problem definition, variants, and analysis (\Cref{sec:why-not}).

\subsection*{Organization}
In brief, the manuscript is organized as follows. After preliminary definitions in \Cref{sec:prelims}, we introduce and carefully motivate our framework in \Cref{sec:shap}, where we also indicate how it could be used for the problems mentioned above. In \Cref{sec:ind-general}, we present some general insights on the complexity of the associated computational problem. We focus on the special case of conjunctive queries and independent parameters in \Cref{sec:ind-cqs}. We discuss the extension to filters and why-not questions in \Cref{sec:filters}. In \Cref{sec:cor-approx}, we discuss the complexity of approximate computation and correlated parameters. We conclude the manuscript in \Cref{sec:conclusions}. 
\section{Preliminaries}\label{sec:prelims}
We begin with basic terminology and notation that we use throughout the manuscript.

\subsection{General Notation} 
We write $\powerset{S}$ for the power set of $S$. Vectors, tuples and sequences are denoted by boldface letters. If $\vec x = (x_i)_{ i \in I }$ and $J \subseteq I$, then $\vec x_{J} = (x_j)_{j \in J}$. Moreover, $\len{ \vec x } = \card{ I }$ is the number of entries of $\vec x$. We let $[n] = \set{1,\dots,n}$. Truth values are denoted by $\T$ and $\F$. Whenever $x$ refers to an input for a computational problem, we write $\enc{x}$ to refer to the encoding length of $x$ as represented in the input.

A (discrete) \emph{probability distribution} is a function $\Gamma \from S \to [0,1]$, where $S$ is a non-empty countable set, and $\sum_{ s \in S } \Gamma(s) = 1$. The \emph{support} of $\Gamma$ is $\supp(\Gamma) = \set{ s \in S \with \Gamma(s) > 0 }$. We use $\Pr$ for generic probability distributions and, in particular, $\Pr_{ X \sim \Gamma }$ to refer to probabilities for a random variable or random vector $X$ being drawn from the distribution $\Gamma$. Likewise, $\E_{ X \sim \Gamma }$ refers to the expectation operator with respect to the distribution $\Gamma$ of $X$.

\subsection{Schemas and Databases} 
A \emph{relation schema} is a sequence $\vec A = (A_1,\dots,A_k)$ of distinct \emph{attributes} $A_i$, each with an associated \emph{domain} $\dom(A_i)$ of values. We call $k$ the \emph{arity} of $\vec A$. If $\vec A = ( A_1, \dots, A_k )$ is a relation schema and $S = \set{i_1, \dots, i_{\ell} } \subseteq [k]$ with $i_1 < i_2 < \dots < i_{\ell}$, then $\vec A_{S}$ denotes the relation schema $(A_{i_1},\dots,A_{i_{\ell}})$.

A \emph{tuple} over $\vec A$ is an element $\vec a = (a_1,\dots,a_k) \in \dom(A_1) \times \dots \times \dom(A_k)$. The set of tuples over $\vec A$ is denoted by $\Tup[ \vec A ]$. A \emph{relation} over $\vec A$ is a finite set of tuples over $\vec A$. We denote the space of relations over $\vec A$ by $\Rel[ \vec A ]$. We typically denote relations by $T$, possibly with a subscript or superscript. \emph{By default, we assume that all attribute domains are countably infinite.}

A \emph{database schema} is a finite set of \emph{relation symbols}, where every relation symbol is associated with a relation schema. 
For example, if $R$ is a relation symbol with associated relation schema $\vec A = (A_1,\dots,A_k)$, we may refer to $R$ by $R(\vec A)$ or $R(A_1,\dots,A_k)$, and call $k$ the \emph{arity} of $R$. 
A \emph{fact} $f$ over a database schema $\vec S$ is an expression of the form $R(\vec a)$, where $R$ is a relation symbol of $\vec S$ and $\vec a$ is a tuple over the relation schema $\vec A$ associated with $R$. 
A \emph{database} over the schema $\vec S$ is a finite set of facts over $\vec S$. We denote the space of databases over some schema $\vec S$ by $\DB[ \vec S ]$. The active domain $\adom(D)$ of a database $D$ is the (finite) set of all values occurring in the facts of $D$.

\subsection{Queries and Parameters}
A \emph{relational query} (or just \emph{query}) is a function that maps databases to relations. More precisely, a query $q$ has a database schema $\inschema(q)$ for its valid input databases (called the \emph{input schema}), and a relation schema $\outschema(q)$ for its output relation (called the \emph{output schema}). Then $q$ is a function that maps databases over $\inschema(q)$ to relations over $\outschema(q)$. We write $q = q( \vec R )$ to denote that $\outschema(q) = \vec R$.

We focus on queries that are expressed in variants of first-order logic without equality. 
In this case, the attributes of $\outschema(q)$ are the free variables of $q$, listed in a fixed order.
\emph{Whenever we discuss queries in this paper, it is assumed that they are queries in first-order logic (FO), unless explicitly stated otherwise. Moreover, by renaming bound variables if necessary, we always assume that no variable appears both free and bound by a quantifier in the same query.} 
We can then write $q = q(\vec x)$ where 
$\vec x = (x_1,\dots,x_n)$ lists the free variables of $q$. 
If $\vec a = (a_1, \dots, a_n) \in \Tup[ \vec R ]$, then $q(\vec a)$ denotes the Boolean query obtained from $q$ by removing all free variables from the head of $q$ and substituting every occurrence of $x_i$ with $a_i$ for all $i \in [n]$.
We write $D \models \varphi$ to denote that the database $D$, viewed as a finite relational structure, satisfies the first-order sentence $\varphi$ under the standard semantics of first-order logic.
If $D$ is a database over $\inschema(q)$, then $q(D) \coloneqq \set{ \vec a \in \Tup[\vec R] \with D \models q( \vec a ) }$; that is, the set of tuples $\vec a$ such that $q(\vec a)$ evaluates to $\mathsf{true}$ on $D$.

A \emph{parameterized query} $Q$ is a relational query in which some of the free variables are distinguished as \emph{parameters}. We then write $\pquery{Q}{\vec x}{\vec y}$ to denote that $\vec x$ are the non-parameter (free) variables of $Q$, and that $\vec y$ are the parameters of $Q$.
A parameterized query $Q$ has a \emph{parameter schema} $\parschema(Q)$, such that $\Tup[\parschema(Q)]$ is the set of valid tuples of parameter values for $Q$. If $Q = \pquery{Q}{\vec x}{\vec y}$ is a parameterized query, and $\vec p \in \Tup[ \parschema(Q) ]$, then $\pqueryinst{Q}{\vec p}(\vec x) \coloneqq \pquery{Q}{\vec x}{\vec p}$ is the (non-parameterized) relational query obtained from $Q$ by removing the parameters $\vec y$ from the head of $Q(\vec x, \vec y)$ and substituting them with the constants $\vec p$. 
The output schema $\outschema(Q)$ of $Q$ is the relation schema whose attributes are the non-parameter free variables $\vec x$ and the input schema $\inschema(Q)$ is that of $Q$. Substituting constants for the parameters removes exactly the parameter variables $\vec y$ and leaves $\vec x$ as the free variables; hence $\inschema(Q) = \inschema(Q_{\vec p})$ and $\outschema(Q) = \outschema(Q_{\vec p})$, which is independent of the choice of $\vec p \in \Tup[ \parschema(Q) ]$.
By this definition, in particular, parameter values are \emph{not} part of the query output, and a parameterized query is considered Boolean if it has no non-parameter free variables.

The \emph{input} and \emph{output schemas} of $Q$ coincide with $\inschema(\pqueryinst{Q}{\vec p})$ and $\outschema(\pqueryinst{Q}{\vec p})$, respectively, which is independent of the choice of $\vec p \in \Tup[ \parschema(Q) ]$.
By this definition, in particular, parameter values are \emph{not} part of the query output, and a parameterized query is considered Boolean, if it has no non-parameter free variables.

We write $Q = \pquery{Q}{\vec R}{\vec P}$ to denote that $\outschema(Q) = \vec R$ and $\parschema(Q) = \vec P$.
If $Q = \pquery{Q}{\vec R}{\vec P}$, then for every database $D$ over $\inschema(Q)$ and every $\vec p \in \Tup[ \vec P ]$ we have
\begin{equation}\label{eq:Q_p}
    \pqueryinst{Q}{\vec p}(D) \coloneqq 
    \set[\big]{ \vec a \in \Tup[ \vec R ] \with D \models \pqueryinst{Q}{\vec p}(\vec a)  } \in \Rel[\vec R]\text.
\end{equation}

A \emph{conjunctive query (CQ)} is a query of the form
	\[
	q(\vec x) = \exists \vec z \with R_1(\vec u_1) \wedge \dots \wedge R_m(\vec u_m)\text,
	\]
	where $R_1,\dots,R_m$ are relation symbols of $\inschema(q)$, and each $\vec u_j$ is a tuple of variables from $\vec x$ and $\vec z$ and of constants, such that every variable of $\vec x$ occurs in some $\vec u_j$. We call the $R_j(\vec u_j)$ the \emph{atoms} of $q$, and we call $q$ \emph{full} if $\vec z$ is empty. Likewise, a \emph{parameterized CQ} is a parameterized query of the form $\pquery{Q}{\vec x}{\vec y} = \exists \vec z \with R_1(\vec u_1) \wedge \dots \wedge R_m(\vec u_m)$ as above, except that the tuples $\vec u_j$ may additionally contain parameters from $\vec y$. We call $Q$ a \emph{full parameterized CQ} if $\vec z$ is empty, that is, if $Q$ contains no existentially quantified variables.

\begin{exa}\label{exa:flights}
    For illustration, consider a parameterized query over flights data, assuming access to a database with relations $\mathsf{Flights}(\mathsf{Id},\mathsf{Date}, \mathsf{AirlineName}, \mathsf{From}, \mathsf{To}, \mathsf{Departure}, \mathsf{Arrival})$ and $\mathsf{Airline}(\mathsf{Name},\mathsf{CountryOfOrigin})$. 
    Then
    \[
        \pquery{Q}{ x, t_{\mathsf{dep}}, t_{\mathsf{arr}}}{ d, c } = 
        \exists a\colon \mathsf{Flights}( x, d, a, \mathtt{CDG}, \mathtt{JFK}, t_{\mathsf{dep}}, t_{\mathsf{arr}}) \wedge \mathsf{Airline}( a, c )
    \]
    is a parameterized conjunctive query. For parameters $d$ and $c$, it asks for departure and arrival times $t_{\mathsf{dep}}$ and $t_{\mathsf{arr}}$ of flights $x$ from Paris ($\mathtt{CDG}$) to New York City ($\mathtt{JFK}$) on date $d$ operated by an airline from country $c$.

    To highlight a possible choice of parameters, consider $d_0 = \mathtt{2025\texttt-01\texttt-31}$ and $c_0 = \mathtt{France}$. Then
    \[
        \pqueryinst{Q}{d_0,c_0}(x,t_{\mathsf{dep}},t_{\mathsf{arr}})
        =
        \exists a\colon \mathsf{Flights}( x, \mathtt{2025\texttt-01\texttt-31}, a, \mathtt{CDG}, \mathtt{JFK}, t_{\mathsf{dep}}, t_{\mathsf{arr}}) \wedge \mathsf{Airline}( a, \mathtt{France})\text.
        \qedhere
    \]
    
\end{exa}

\subsection{The Shapley Value}\label{sec:pre-shapleyvalue}

A \emph{cooperative game} is a pair $(I,\nu)$ where $I$ is a non-empty set of players, and $\nu \colon \powerset{I} \to \RR$ is a \emph{utility function} that assigns a ``joint wealth'' $\nu(J)$ to every \emph{coalition} $J \subseteq I$. Without loss of generality, and unless explicitly stated otherwise, we will consistently assume that $I = [\ell]$ for some $\ell \in \mathbb{N}$.

The \emph{Shapley value} of player $i \in [\ell]$ in a cooperative game $([\ell],\nu)$ measures the importance (or, contribution) of player $i$ among coalitions. Formally, it is defined as
\begin{equation}\label{eq:shapleyv1}
    \Shapley({[\ell]},\nu,i) \coloneqq
    \frac{1}{{\ell}!} \sum_{ \sigma \in S_{{[\ell]}} } \bigl( \nu( \sigma_i \cup \set{i} ) - \nu( \sigma_i ) \bigr) =
    \E_{ \sigma \sim S_{{[\ell]}} }\bigl[ \nu( \sigma_i \cup \set{i} ) - \nu( \sigma_i ) \bigr]
    \text,
\end{equation}
where $S_{{[\ell]}}$ is the set of permutations of {$[\ell]$}, and $\sigma_i$ is the set of players appearing before $i$ in the permutation $\sigma \in S_{{[\ell]}}$, and $\sigma \sim S_{{[\ell]}}$ {indicates that $\sigma$ is drawn according to} the uniform distribution on $S_{{[\ell]}}$. 
An equivalent, alternative formulation is
\begin{equation}
    \label{eq:shapleyv2exp}
    \Shapley({[\ell]},\nu,i) =
    \sum_{J \subseteq {[\ell]}\setminus \set{i}} \Pi_i(J) \cdot \bigl( \nu( J \cup \set{i} ) - \nu( J ) \bigr) =
    \E_{ J \sim \Pi_i }\bigl[ \nu( J \cup \set{i} ) - \nu( J ) \bigr]\text,
\end{equation}
where $\Pi_i$ is the probability distribution on $\powerset{{[\ell]} \setminus \set{i}}$ given by
\begin{equation}\label{eq:subsetdist}
    \Pi_i( J ) = 
    \frac{ \card{J}! \cdot ( {\ell} - \card{J} - 1 )! }{ {\ell}! } = 
    \frac{ 1 }{ {\ell} \cdot \binom{ \ell - 1 }{ \card{J} }}\text.
\end{equation}

As proved by Shapley~\cite{Shapley}, the Shapley value is the \emph{only} way of quantifying the contribution of the players, subject to a few natural assumptions (``axioms'') and the technical requirement that $\nu$ is zero on the empty set.\footnote{We drop the requirement that $\nu$ is zero on the empty set throughout the manuscript. This is commonly done in literature on the $\Shap$ score, since the requirement can be achieved by considering $\nu'(J) \coloneqq \nu(J) - \nu(\emptyset)$ instead of $\nu$ and both valuation functions $\nu$ and $\nu'$ yield the same Shapley values.} Let us briefly review these axioms. (For details and a proof of Shapley's theorem, we refer to the literature, e.g., \cite{Shapley,roth1988shapley,algaba2019handbook}.) The first axiom, \emph{symmetry}, says that if any two players $i$ and $j$ have the same contribution to any coalition $J \subseteq [\ell]\setminus\set{i,j}$, then $\Shapley([\ell],\nu,i) = \Shapley([\ell],\nu,j)$. The second axiom, \emph{efficiency}, says that the sum of $\Shapley([\ell],\nu,j)$ for all players $j \in [\ell]$ is $\nu([\ell])-\nu(\emptyset)$, that is, the contribution of the ``grand coalition'' $[\ell]$ minus the contribution of the empty coalition $\emptyset$. 
Thus, the Shapley values ``distribute'' $\nu([\ell])$ among the individual players. The third axiom, \emph{null player}, says that if a player $j$ makes no contribution to any coalition, then $\Shapley([\ell],\nu,j) = 0$. The fourth axiom, \emph{additivity}, says that if we have two different functions $\nu_1,\nu_2 : 2^{[\ell]} \to \RR$ measuring the value of a coalition, then $\Shapley([\ell],\nu_1 + \nu_2,j) = \Shapley([\ell],\nu_1,j) + \Shapley([\ell],\nu_2,j)$.

\section{The \texorpdfstring{$\Shap$}{SHAP} Score of Query Parameters}\label{sec:shap}

In this section, we give our exact definition of the $\Shap$ score of a query parameter, and argue that it measures in a natural way the contribution of the parameter to the outcome of the query. We start with a database $D$ over $\inschema(Q)$; a parameterized query $\pquery{Q}{\vec{R}}{\vec{P}}$; and a \emph{reference parameter} $\vec p^* = (p_1^*,\dots,p_\ell^*)$ over $\vec P$. We want to quantify the contribution of each parameter $p_i^*$ to the query answer $Q_{\vec p^*}(D)$.

\subsection{Intuition}

A basic idea is to see what happens to the answer if we either modify $p_i^*$ and keep all other parameters fixed or, conversely, keep $p_i^*$ fixed but modify the other parameters. The following example shows that neither of these two approaches is sufficient.

\begin{exa}\label{exa:coalition}
    We let $Q(x;y_1,y_2,y_3)\coloneqq R(y_1,y_2,y_3,x)$ for  a relation $R(B_1,B_2,B_3,A)$. Let $D$ be the database where the tuple set of $R$ is
    \[
        \bigl( \bigl(\set{1}\times [n] \bigr) \cup \bigl( [n]\times\set{1} \bigr) \bigr) \times [n] \times [n]\text.
    \]
    Let $\vec p^*=(p_1^*,p_2^*,p_3^*)=(1,1,1)$. Then $Q_{\vec p^*}(D)=[n]$.

    Let us try to understand the contribution of each parameter $p_i^*$ when we only consider parameter values from $[n]$. Intuitively, the parameter $p_3^*$ is completely irrelevant. Whatever value it takes in any configuration of the other two parameters, the query answer remains the same. Parameters $p_1^*$ and $p_2^*$ are important, though. If we change both of their values to values $p_1,p_2\in [n]\setminus\set{1}$, the query answer becomes empty. However, for this change to take effect, \emph{we need to change both values at the same time}. 
    If we change only $p_1$ keeping $p_2^*$ and $p_3^*$ fixed, the query answer does not change. This is also the case, if we keep $p_1^*$ fixed and modify $p_2$ and $p_3$. The same holds for the second parameter. 
\end{exa}

What this example shows is that even if we just want to understand the contribution of the individual parameters, we need to look at the contributions of \emph{sets} of parameters. This will immediately put us in the realm of \emph{cooperative game theory} with the parameters acting as players in a coalitional game:
Assume for a moment that we have a function $\nu : 2^{[\ell]} \to \RR$ such that, for $J \subseteq [\ell]$, the function value $\nu(J)$ quantifies the combined impact of the parameters with indices in $J$. (We discuss in the next subsection how to obtain such a function.)
Given $\nu$, we now quantify the contribution of every \emph{individual} parameter $j \in [\ell]$. This is exactly what the Shapley value $\Shapley([\ell],\nu,j)$ does.

Concluding the discussion so far, to understand the contribution of individual parameters, we need to look at sets of parameters, and the method of choice to retrieve this contribution from a valuation for sets is the Shapley value.

\subsection{The Utility Function}

What remains to be done is constructing a natural valuation for sets of parameters, that is, a function $\nu:2^{[\ell]}\to\mathbb R$ such that, for $J\subseteq[\ell]$ the value $\nu(J)$ quantifies the combined contribution of the parameters $p_j^*$, for $j\in J$, to the query answer. 
This requires a way of measuring the similarity between query answers, that is, a function $\simi$ that takes two relations of the schema of our query answer and assigns a similarity value to them. Formally, a \emph{similarity function} is just an arbitrary function $\simi:\Rel[\vec R]\times\Rel[\vec R]\to\mathbb R$.
We make no restrictions on this function, similarity can have different meanings depending on the application. Let us just note that the intended use of this function is always to quantify similarity: the higher $\simi(Q,Q^*)$, the more similar $Q$ and $Q^*$ are.

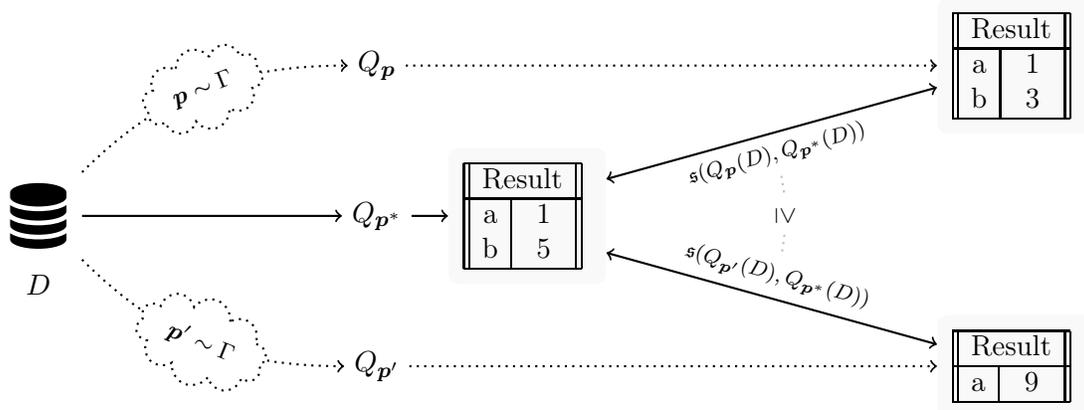
\begin{figure}
    \centering
\colorlet{bgshade}{lightgray!10!white}
\begin{tikzpicture}
    \node[align=center,label={below:$D$},rectangle,rounded corners,inner sep=2mm,font=\Huge] (DB) at (-1,0) {
        \faDatabase
    };
    
    \node[anchor=center] (Qp*) at (3.5,0) {$Q_{ \bm p^* }$};
    \node[anchor=center] (Qp1) at (3.5,2) {$Q_{ \bm p }$};
    \node[anchor=center] (Qp2) at (3.5,-2) {$Q_{ \bm p' }$};
    
    \node[align=center,fill=bgshade,rectangle,rounded corners,inner sep=2mm] (resp*) at (5.5,0) {
        \begin{tabular}{||c|c||}\hline
            \multicolumn{2}{||c||}{Result}\\\hline
            a & 1 \\
            b & 5\\\hline
        \end{tabular}
    };
    
    \node[align=center,fill=bgshade,rectangle,rounded corners,inner sep=2mm] (resp1) at (12,2) {
        \begin{tabular}{||c|c||}\hline
            \multicolumn{2}{||c||}{Result}\\\hline
            a & 1\\
            b & 3 \\\hline
        \end{tabular}
    };

    \node[align=center,fill=bgshade,rectangle,rounded corners,inner sep=2mm] (resp2) at (12,-2) {
        \begin{tabular}{||c|c||}\hline
            \multicolumn{2}{||c||}{Result}\\\hline
            a & 9\\\hline
        \end{tabular}
    };
    
    \draw[thick,->] (DB) to (Qp*.west);
    \draw[thick,->,dotted] (DB) to[out=45,in=180]  node[midway,cloud,draw,fill=white,aspect=2,sloped,font=\footnotesize] {$\bm p\sim\Gamma$}  (Qp1.west);
    \draw[thick,->,dotted] (DB) to[out=-45,in=180] node[midway,cloud,draw,fill=white,aspect=2,sloped,font=\footnotesize] {$\bm p'\sim\Gamma$} (Qp2.west);
    \draw[thick,->] (Qp*.east) to (resp*);
    \draw[thick,->,dotted] (Qp1.east) to (resp1);
    \draw[thick,->,dotted] (Qp2.east) to (resp2);
    \draw[thick,<->] (resp*.25) to node[below,sloped,font=\scriptsize] (compare1) {$\simi( \pqueryinst{Q}{\vec p}(D), \pqueryinst{Q}{\vec p^*}(D) )$} (resp1);
    \draw[thick,<->] (resp*.-25) to node[above,sloped,font=\scriptsize] (compare2) {$\simi( \pqueryinst{Q}{\vec p'}(D), \pqueryinst{Q}{\vec p^*}(D) )$} (resp2);

    \draw[thick,lightgray,dotted] (compare1) to[bend left=10] node[midway,fill=white,text=black,sloped,font=\footnotesize] {$\geq$} (compare2);
\end{tikzpicture}
    \caption{On a given database $D$, we evaluate a parameterized query $Q$ with respect to a reference parameter $\vec p^*$. The result is compared (using a similarity function $\simi$) against different parameterizations of $Q$, according to a probabilistic model $\Gamma$.}
    \label{fig:illustration}
\end{figure}

\begin{exas}\label{exa:similarityfunctions}
    \begin{enumerate}[(a)]
    \item\label{itm:setbasedliterature} Prime examples of similarity functions are set-similarity measures such as the \emph{Jaccard index}, the \emph{S\o{}rensen index}, and \emph{Tversky's index}, which attempt to capture the degree of similarity between two sets \cite{ontanon2020overview}. If $Q$ and $Q^*$ are two relations of the same relation schema $\vec R$, then, e.g., their Jaccard index is given by
        \[
            \Jaccard(Q,Q^*) \coloneqq
            \frac{ \card{ Q \cap Q^* } }{ \card{ Q \cup Q^* } } =
            1 - \frac{ \card{ Q \mathbin{\triangle} Q^* } }{ \card{ Q \cup Q^* } }\text,
        \]
        and $\Jaccard( \emptyset, \emptyset ) = 1$ by convention. 
    \item\label{itm:setbasedsimple} In the same spirit, but somewhat simpler, are similarities based on cardinalities and basic set operations, for example:
        \begin{align*}
            \Intersection(Q,Q^*)    &   \coloneqq \card{ Q \cap Q^* }\text,&\text{(size of the intersection)}\\          
            \NegSymDiff(Q,Q^*)      &   \coloneqq -\card{ Q \mathbin{\triangle} Q^* }\text,&\text{(negative symmetric difference)}\\
            \NegSymCDiff(Q,Q^*)        &   \coloneqq-\card[\big]{ \card{Q}-\card{Q^*} }\text.&\text{(negative symmetric cardinality difference)}
        \end{align*}
        Note that we do not require similarity to be non-negative, thus both $\NegSymDiff$ and $\NegSymCDiff$ are well-defined. Intuitively, the negative sign makes sense because two relations get more similar as their symmetric difference gets smaller.
    \item\label{itm:setbasedasym} In our applications, $Q$ and $Q^*$ have different roles: $Q^*$ is always a fixed ``reference relation'', and $Q$ is a modification of $Q^*$, obtained by changing some parameters in the query. In some situations, it can be useful to use similarity measures that are not symmetric. For example, we may only care about not losing any tuples from $Q^*$, whereas additional tuples may be irrelevant. This leads to an asymmetric similarity measure:
       \begin{align*}
            \NegDiff(Q,Q^*) & \coloneqq -\card{ Q^* \setminus Q }\text. & \text{(negative difference)}
        \end{align*}
\end{enumerate}
        
We can also use similarity functions that are specialized to the application at hand. 
For example, we can use a similarity function that accounts for differences of attribute values, such as, e.g., arrival time minus departure time:
\[
    \simifont{MinDiff}_{A,B}(Q,Q^*) \coloneqq - \abs[\big]{ \min(Q.A - Q.B) - \min(Q^*.A - Q^*.B)}\text,
\]
where $Q.A$ is the vector obtained from the numerical column $A$ of $Q$ (hence, $Q.A - Q.B$ is the vector of differences between the $A$ and $B$ attributes). Alternatively, we can use 
\[
    \simifont{ExpMinDiff}_{A,B}(Q,Q^*) \coloneqq \exp\bigl( \simifont{MinDiff}_{A,B}(Q,Q^*) \bigr)\text,
\]
which takes values in $(0,1]$.
\end{exas}

Suppose now that we have a subset $J\subseteq[\ell]$. We need a value $\nu(J)$ quantifying how important the subtuple $\vec p^*_J=(p_j)_{j\in J}$ is to obtain a query answer similar to $Q_{\vec p^*}(D)$. The crucial idea of the $\Shap$ score in a very similar setting in machine learning~\cite{lund17} was to see what happens if we fix the parameters $\vec p^*_J$ and change the other parameters: the value of $J$ is high if changing the other parameters has no big effect on the query answer. In other words, the value of $J$ is large if for parameter tuples $\vec p$ with $\vec p_J=\vec p^*_J$ the relations $Q_{\vec p}(D)$ and $Q_{\vec p^*}(D)$ are similar, that is, $\simi(\pqueryinst{Q}{\vec p}(D),\pqueryinst{Q}{\vec p^*}(D))$ is large. We use the shorthand notation \[
\simi(\vec p,\vec p^*) = \simi(\pqueryinst{Q}{\vec p}(D),\pqueryinst{Q}{\vec p^*}(D))
\]
in case that $Q$ and $D$ are clear from the context.
But this, of course, depends on how exactly we change the other parameters. The idea is to change them randomly and take the expected value, as illustrated in \autoref{fig:illustration}. So, we assume that we have a probability distribution $\Gamma$ on the space of all possible parameter values, and we define
\begin{equation}\label{eq:def-nu-SHAP}
    \nu(J) 
    \coloneqq \E_{\vec p\sim\Gamma}\big[ \simi(\vec p,\vec p^*) \;\big|\; \vec p_J=\vec p^*_J\big]
    \text.
\end{equation}

With the above, we can quantify the contribution of any individual parameter $i$ as $\Shapley\big([\ell],\nu,i\big)$. This contribution depends on~$\Gamma$: for the same $Q$, $D$, $\vec p^*$, and $\simi$, different choices of $\Gamma$ may lead to very different scores, since $\Gamma$ specifies the alternatives against which the actual parameter values are assessed. By using different probability distributions, we can thus incorporate various assumptions about the usage of parameters and answer different, equally legitimate questions. For example, $\Gamma$ may apply to only a subset of interesting parameters and be deterministic (fixed) on the others; it can also be concentrated on certain values (e.g., locations of a reasonable distance to the original value) to measure the robustness of the answer under small perturbations, as in fact checking~\cite{DBLP:journals/tods/0001ALYY17}, or be broad, such as the empirical distribution of the attribute values in~$D$, to measure global influence; in general, $\Gamma$ can describe any statistical model on the parameter space.

There is a different, in some way a complementary, approach to quantify the contribution of a set $J$ of parameters: we can fix the other parameters and randomly change those in~$J$; then the more the query answer changes, the more important $J$ is. Using the dissimilarity $\overline\simi(\vec p,\vec p') \coloneqq c - \simi(\vec p,\vec p')$ for an arbitrary constant $c \in \RR$,\footnote{The constant $c$ does not affect the Shapley values; for similarity functions with values in $[0,1]$, such as $\Jaccard$, $c = 1$ is a natural choice. We start with $\nu$ because its analogue is the one used in machine learning (e.g.,~\cite{DBLP:conf/cascon/MokhtariHB19,DBLP:conf/sibgrapi/MarcilioE20,DBLP:journals/ai/BaptistaGH22,DBLP:journals/frai/GramegnaG21,KIM2022103677}).} and letting $\overline J \coloneqq [\ell] \setminus J$, this leads to the utility function
\[
\overline\nu(J)\coloneqq\E_{\vec p\sim\Gamma}\big[\overline\simi(\vec p,\vec p^*) \;\big|\; \vec p_{\overline J}=\vec p^*_{\overline J}\big].
\]
Both utility functions are natural, and, fortunately, they lead to the same Shapley values.

\begin{thm}
\label{thm:same_measure}
For all $i\in[\ell]$ we have $\Shapley\big([\ell],\nu,i\big)=\Shapley\big([\ell],\overline\nu,i\big)$.
\end{thm}

\begin{proof}
From the linearity of expectation we get
\begin{align}
\nu(J)
    &=\E_{\vec p\sim\Gamma}\big[\simi(\vec p,\vec p^*) \;\big|\; \vec p_{ J}=\vec p^*_{J}\big]
    =\E_{\vec p\sim\Gamma}\big[c-\overline\simi(\vec p,\vec p^*) \;\big|\; \vec p_{J}=\vec p^*_{J}\big] \notag
\\
    &= c-\E_{\vec p\sim\Gamma}\big[ \simi(\vec p,\vec p^*) \;\big|\; \vec p_{J}=\vec p^*_{J}\big]
    =c-\overline\nu(\overline{J})\label{eq:simi-vs-dissimi}
\end{align}
By applying \eqref{eq:shapleyv2exp}, we have
\begin{align*}
\Shapley\big([\ell],\nu,i\big) &=
    \sum_{{J} \subseteq [\ell]\setminus \set{i}}
        \Pi_i(J) \cdot \bigl(
            \nu( J \cup \set{i} ) - \nu( J )
        \bigr)
   \\
   & =
    \sum_{{J} \subseteq [\ell]\setminus \set{i}}
        \Pi_i(J) \cdot \bigl(
            c-\overline\nu(\overline{J\cup\set{i}}) - c + \overline\nu(\overline{J})
        \bigr)\\
         &=
    \sum_{{J} \subseteq [\ell]\setminus \set{i}}
        \Pi_i(J) \cdot \bigl(
           \overline\nu(\overline{J})-\overline\nu(\overline{J\cup\set{i}})
        \bigr)\text.
\end{align*}
Letting $J'=\overline{J\cup\set{i}}$ and observing that
\[
\Pi_i(J)=\frac{ \card{J}! \cdot (\ell - \card{J} - 1 )! }{ \ell! }=\frac{(\ell-\card{J'}-1)!\cdot\card{J'}!}{\ell!}=\Pi_i(J')\text,
\]
we get
\[
\Shapley\big([\ell],\nu,i\big)=
\sum_{{J} \subseteq [\ell]\setminus \set{i}}
        \Pi_i(J') \cdot \bigl(
           \overline\nu(J'\cup\set{i})-\overline\nu(J')
        \bigr)\text.
\]
Finally, observe that summing up over $J$ is the same as summing up over $J'$, and, thus,
\[
\Shapley\big([\ell],\nu,i\big)=
\sum_{{J'} \subseteq [\ell]\setminus \set{i}}
        \Pi_i(J') \cdot \bigl(
           \overline\nu(J'\cup\set{i})-\overline\nu(J')
        \bigr) = \Shapley\big([\ell],\overline\nu,i\big)\,,
\]
as claimed.
\end{proof}

Thus, the $\Shap$ score is the same whether importance is measured by the \emph{sufficiency} of parameter values for preserving the answer ($\nu$) or by their \emph{necessity} for it ($\overline\nu$), although these perspectives can disagree drastically for single parameters (in \cref{exa:coalition}, $p_1^*$ alone is sufficient but not necessary); we view this as evidence that it is a natural measure.

\begin{defi}
    For a given parameterized query $Q$, a database $D$ over $\inschema(Q)$, a reference parameter $\vec p^*$ over $\parschema(Q)$, a similarity function $\simi$ and a probability distribution $\Gamma$ over $\Tup[\parschema(Q)]$ with $\Gamma(\vec p^*) > 0$, the \emph{$\Shap$ score of parameter $i$} is given as
    \[
        \Shap_{Q,D,\vec p^*,\simi,\Gamma}(i) = 
        \Shapley( [\ell], \nu, i )\text,
    \]
    where $\ell = \lvert \vec p^* \rvert$ is the number of parameters in $Q$, $\nu$ is defined as in  \eqref{eq:def-nu-SHAP}, and $i \in [\ell]$.
\end{defi}

The full expression for the $\Shap$ score of parameter $i$ is, hence,
\begin{align*}
    &\Shap_{Q,D,\vec p^*,\simi,\Gamma}(i) = \\
    &\E_{ J \sim \Pi_i }\Bigl[
        \E_{\vec p\sim\Gamma}\big[ \simi(\pqueryinst{Q}{\vec p}(D),\pqueryinst{Q}{\vec p^*}(D)) \;\big|\; \vec p_{J\cup\set{i}}=\vec p^*_{J\cup\set{i}}\big] -
        \E_{\vec p\sim\Gamma}\big[ \simi(\pqueryinst{Q}{\vec p}(D),\pqueryinst{Q}{\vec p^*}(D)) \;\big|\; \vec p_J=\vec p^*_J\big]
    \Bigr]
    \text.
\end{align*}
When $Q,D,\vec p^*,\simi,\Gamma$ are clear from the context, we write $\Shap(i)$ instead of $\Shap_{Q,D,\vec p^*,\simi,\Gamma}(i)$.
We illustrate this framework using our previous examples.

\begin{exa}
    Continuing \autoref{exa:flights}, we assess the importance of $Q$'s parameters with respect to the reference parameter $\vec p^* = (d_0, c_0)$ with $d_0 = \mathtt{2025\texttt-01\texttt-31}$ and $c_0 = \mathtt{France}$. The distribution $\Gamma$ lets us tune which alternatives we consider plausible: to focus on local perturbations, we may take $\Gamma$ to be a product distribution supported on dates near $d_0$ and on the countries $\set{\mathtt{France}, \mathtt{USA}}$. The choice of similarity function then determines what is being compared. With $\NegSymCDiff$ we compare $\pqueryinst{Q}{\vec p}$ and $\pqueryinst{Q}{\vec p^*}$ by their number of flight options; with $\NegDiff$, we measure how many of the reference’s flight options are lost when the parameters change; and with $\simifont{MinDiff}_{\mathsf{Arrival},\mathsf{Departure}}$, we compare the durations of the shortest available flights.
\end{exa}

Now that we have illustrated the intuition, we revisit \autoref{exa:coalition} to show how $\Shap$ scores are computed.

\begin{exa}
    In \autoref{exa:coalition}, we considered the query $\pquery{Q}{x}{y_1,y_2,y_3} = R(y_1,y_2,y_3,x)$ on a database $D$ in which the relation $R$ contains the tuples $(b_1,1,b_3,a)$ and $(1,b_2,b_3,a)$ for all $a,b_1,b_2,b_3 \in [n]$.
    For simplicity, we discuss the case $n=2$. Then the relation $R$ in $D$ has the following tuples:
    \begin{align*}
        \underline{(1,1,1,1)}, (1,2,1,1), (2,1,1,1),&&
        \underline{(1,1,1,2)}, (1,2,1,2), (2,1,1,2),\\
        (1,1,2,1), (1,2,2,1), (2,1,2,1),&&
        (1,1,2,2), (1,2,2,2), (2,1,2,2)\text.
    \end{align*}
    We had already picked a reference parameter $\vec p^* = (1,1,1)$, which yields $\pqueryinst{Q}{\vec p^*}(D) = \set{1,2}$, coming from the two underlined tuples above. In using the similarity function $\NegDiff$ from \autoref{exa:similarityfunctions}, we measure how many output values are lost when using some different $\vec p$ instead of $\vec p^*$. Observe that for any $\vec p = (p_1,p_2,p_3) \in \set{1,2}^3$, we get $\pqueryinst{Q}{\vec p}(D) = \set{1,2}$ if $(p_1,p_2) \neq (2,2)$, and $\pqueryinst{Q}{\vec p}(D) = \emptyset$ otherwise. Hence,
    \[
        \NegDiff(\vec p,\vec p^*) = 
        \begin{cases}
             0 & \text{if }(p_1,p_2) \neq (2,2)\text,\\
            -2 & \text{if }(p_1,p_2) = (2,2)\text.
        \end{cases}
    \]
    A simple calculation then shows
    \begin{align*}
        \nu(J) 
        &= \E_{\vec p \sim \Gamma}\bigl[ \NegDiff(\vec p,\vec p^*) \mid \vec p_J = \vec p_J^* \bigr]\\
        &= 0 \cdot \Pr_{\vec p \sim \Gamma}\bigl( \NegDiff(\vec p,\vec p^* ) = 0 \mid \vec p_J = \vec p_J^* \bigr) + (-2) \cdot \Pr_{\vec p \sim \Gamma}\bigl( \NegDiff(\vec p,\vec p^* ) = -2 \mid \vec p_J = \vec p_J^* \bigr)\\
        &= (-2) \cdot \Pr_{\vec p \sim \Gamma}\bigl((p_1,p_2) = (2,2) \mid \vec p_J = \vec p_J^* \bigr)
    \end{align*}
    Suppose that $\Gamma$ is the uniform distribution on $\set{1,2}^3$, i.e., each of the three parameters takes values $1$ or $2$ with probability $\frac12$. Then
    \[
        \Pr_{\vec p \sim \Gamma}\bigl((p_1,p_2) = (2,2) \mid \vec p_J = \vec p_J^* \bigr) = \begin{cases}
        0 & \text{if } 1 \in J \text{ or } 2 \in J \\
        \frac{1}{4} & \text{otherwise,}
        \end{cases}
    \]
    and, hence,
    \[
        \nu(J) = 
        \begin{cases}
            -\tfrac12   & \text{if } J= \emptyset \text{ or } J = \set{3}\text,\\
            0           & \text{otherwise.}
        \end{cases}
    \]
    As $\Shap(i) = \E_{J \sim \Pi_i}[ \nu( J\cup \set{i} ) - \nu(J) ]$, we compute
    \begin{align*}
        \Shap(1) &= 
        \begin{aligned}[t]
                    &\tfrac{1}{3} \cdot \bigl(\nu(\set{1}) - \nu(\emptyset)\bigr)   \\
            {}+{}   &\tfrac{1}{6} \cdot \bigl(\nu(\set{1,2}) - \nu(\set{2})\bigr) \\
            {}+{}   &\tfrac{1}{6} \cdot \bigl(\nu(\set{1,3}) - \nu(\set{3})\bigr) \\
            {}+{}   &\tfrac{1}{3} \cdot \bigl(\nu(\set{1,2,3}) - \nu(\set{2,3})\bigr)
        \end{aligned} \\
        & = 
        \tfrac{1}{3} \cdot \bigl(0 - (-\tfrac{1}{2})\bigr)
        + \tfrac{1}{6} \cdot \bigl(0 - 0\bigr)
        + \tfrac{1}{6} \cdot \bigl(0 - (-\tfrac{1}{2})\bigr)
        + \tfrac{1}{3} \cdot \bigl(0 - 0\bigr) \\
        & = \tfrac{1}{4} \text.
    \end{align*}
    Similarly, we obtain $\Shap(2) = \frac14$ and $\Shap(3) = 0$, meaning that the first two parameters are of equal importance, whereas the last parameter has no impact with respect to $\NegDiff$ and $\vec p^* = (1,1,1)$ under the uniform distribution.

    For general $n$ and uniform $\Gamma$, we observe the same behavior for the $\Shap$ scores, as $\Shap(1) = \Shap(2) = \frac{(n-1)^2}{2n}$ and $\Shap(3) = 0$. Instead of going through the whole computation again, we can get this directly by exploiting the symmetry axiom for parameters $y_1$ and $y_2$, the null player axiom for $y_3$, and the observation that $\nu(\emptyset)=(-n)\cdot\frac{(n-1)^2}{n^2}$ and $\nu(\set{1,2,3})=0$, which means that the sum of the Shapley values of the three parameters is $\frac{(n-1)^2}{n}$ by the efficiency axiom. 
\end{exa}

\paragraph*{Potential Applications.}
In machine learning, attribution scores such as $\Shap$ are used for debugging models, for auditing, and for explaining individual decisions~\cite{bhatt2020explainable}, although their suitability for such tasks is debated~\cite{kumar2020problems}. Our scores could serve analogous purposes for the problems mentioned in the introduction, since, by the efficiency axiom, they split the gap $\simi(\vec p^*,\vec p^*) - \E_{\vec p\sim\Gamma}[\simi(\vec p,\vec p^*)]$ between the actual and the typical outcome among the parameters. 
For $\simi(T,T^*) = \card{T}$, this gap is the difference between the actual and the expected number of answers under $\Gamma$, and the parameters with the most negative (positive) scores are candidates to be relaxed (tightened) first.
If $\simi(T,T^*) = 1$ when $T$ supports the same conclusion as $T^*$, and $0$ otherwise, the gap is the probability that the conclusion of a claim flips, and high scores point to potentially cherry-picked parameter values. Similarly, for $\simi(T,T^*) = 1$ if a missing tuple $\vec t$ is in $T$, and $0$ otherwise, the most negative scores indicate which parameter values exclude $\vec t$, complementing the Shapley values of filters in \Cref{sec:why-not}. As in machine learning, whether such uses are effective in practice is an empirical question that we leave to future, more applied research.

\subsection{Formal Problem Statement}

We now have all the ingredients to formulate the computation of $\Shap$ scores as a formal computational problem. It can be phrased concisely as follows.

\begin{prob}\label{prob:shap}
    $\SHAP(\Q,\PR,\simi)$ is the following computational problem: 

    On input $(Q,\vec p^*,D,\Gamma)$, compute $\Shap(i)$ for all $i \in [\ell]$.
\end{prob}

\autoref{tab:problem} presents an overview of the problem specifications and inputs. For technical reasons (which we explain next), the problem is parameterized with a class of parameterized queries $\Q$, a class of parameter distributions $\PR$, and a (class of) similarity function(s) $\simi$.

\begin{table}[bth]
    \centering
    \caption{Problem specifications and inputs for $\Shap$ score computation.}\label{tab:problem}
    \begin{tabular}{llll}\toprule
        \multicolumn{2}{l}{\textbf{\textsf{Problem specifications}}} & 
        \multicolumn{2}{l}{\textbf{\textsf{Problem inputs}}} \\\cmidrule(r){1-2}\cmidrule(l){3-4}
        $\Q$ & class of parameterized queries     & $\pquery{Q}{\vec{R}}{\vec{P}} {\in \Q}$ & parameterized query\\
        $\PR$ & class of parameter distributions  & $D \in \DB[\inschema(Q)]$ & database\\
        $\simi$ & similarity function             & $\vec p^* = (p_1^*,\dots,p_{\ell}^*) \in \Tup[\vec P]$ & reference parameter\\
        &                                         & $\Gamma \in \PR_{\vec P}$, $\Gamma(\vec p^*) > 0$ & parameter distribution\\\bottomrule
    \end{tabular}
\end{table}

\paragraph*{Parameter distributions}
For simplicity, we only consider parameter distributions with finite support and rational probabilities. Parameter distributions are always given by a fixed representation, we write $\|\Gamma\|$ for the encoding length of that representation. A class of parameter distributions will always mean a class 
$\PR = \bigcup_{ \vec P } \PR_{ \vec P }$ where $\vec P$ are relation schemas, and $\PR_{ \vec P }$ is a set of functions $\Gamma : \Tup[ \vec P ] \to [0,1]$ such that $\sum_{ \vec p } \Gamma( \vec p ) = 1$.

One of the simplest classes of parameter distributions is the class $\IND$ of \emph{fully factorized distributions}, where all parameters are stochastically independent.
\begin{defi}
     A probability distribution $\Gamma$ over $\Tup[ \vec P ]$ is called \emph{fully factorized} if
    \[
        \Gamma( p_1, \dots, p_{\ell} ) =
        \Gamma_1( p_1 ) \cdot \ldots \cdot \Gamma_{\ell}( p_{\ell} )
    \]
    for all $( p_1,\dots,p_{\ell} ) \in \Tup[\vec P]$ where $\Gamma_i$ is the marginal distribution of $p_i$ under $\Gamma$. We represent such $\Gamma$ by the $\ell$ lists of all pairs $(p_i,\Gamma_i(p_i))$ with $\Gamma_i(p_i)>0$, where probabilities are given as fractions of binary integers.
\end{defi}

Technically, $\PR$ is a class of \emph{representations} of probability distributions in our definition of the $\SHAP$ problem. For convenience, we do not distinguish between a distribution and its representation.
\paragraph*{Similarity functions}
In \cref{prob:shap}, the definition of the $\SHAP$ problem, we chose to fix the similarity function $\simi$ as part of the problem specification. Different similarity functions greatly differ in their behavior and complexity, and since there is no clear, uniform way of encoding such functions, it is inconvenient to have them as an input to the problem.

For $\SHAP(\Q,\PR,\simi)$, $\simi$ is technically a collection of functions $\simi_{\vec R}$ for every output schema $\outschema(Q)$ for queries in $\Q$, rather than a single function. We later abuse notation and just write $\simi(T_1,T_2)$ instead of $\simi_{\vec R}(T_1,T_2)$ if $T_1$ and $T_2$ are relations over schema $\vec R$. 

\paragraph*{Parameterized queries}
The parameterized queries we consider are formulated in first-order logic, unless explicitly stated otherwise. 
The encoding size of $\pqueryinst{Q}{\vec p}$ can be larger than that of the parameterized query $Q$. By our assumptions on $\Gamma$, this blow-up is at most polynomial in $\enc{Q}$ and $\enc{\Gamma}$ if $\vec p \in \supp(\Gamma)$.

\paragraph*{Complexity}
We abuse notation and use $\simi \after \Q$ to denote the class of functions mapping $(Q, \vec p, \vec p^*, D)$ to $\simi( \pqueryinst{Q}{\vec p}(D), \pqueryinst{Q}{\vec p^*}(D) )$ where $Q = \pquery{Q}{\vec R}{\vec P} \in \Q$, where $D \in \DB[\inschema(Q)]$, and where $\vec p,\vec p^* \in \Tup[\vec P]$.

\begin{defi}
    Let $\Q$ be a class of parameterized queries, $\simi$ a similarity function. We call
    \begin{itemize}
        \item $\Q$ \emph{tractable}, if $(Q,\vec p,D) \mapsto \pqueryinst{Q}{\vec p}(D)$ can be computed in polynomial time.
        \item $\simi$ \emph{tractable}, if $(T_1,T_2) \mapsto \simi( T_1, T_2 )$ can be computed in polynomial time.
        \item $\simi \after \Q$ \emph{tractable}, if $(Q, \vec p, \vec p^*, D) \mapsto \simi( \pqueryinst{Q}{\vec p}(D), \pqueryinst{Q}{\vec p^*}(D) )$ can be computed in polynomial time.
    \end{itemize}
\end{defi}

Tractability of both $\simi$ and $\Q$ implies tractability of $\simi \after \Q$, but the converse is not true in general (see \cite{gurevich1989time}). For illustration of this phenomenon, consider the query answering problem for (non-parameterized) acyclic conjunctive queries without quantification: The query answer may be exponentially large, but the answer \emph{count} can be computed efficiently \cite[Theorem 1]{pichler2013tractable}. 

\begin{defi}
    Let $\Q$ be a class of parameterized queries, $\PR$ a class of parameter distributions, and $\simi$ a similarity function, and let $\mathsf{C}$ be a complexity class.
    \begin{itemize}
        \item  If the complexity of solving $\SHAP(\Q,\PR,\simi)$, as a function of $\enc{Q} + \enc{\vec p^*} + \enc{D} + \enc{\Gamma}$, is in $\mathsf{C}$, then we say that $\SHAP(\Q,\PR,\simi)$ is in $\mathsf{C}$ in \emph{combined complexity}.
        \item  If for every fixed $Q \in \Q$, the complexity of solving $\SHAP( \set{Q}, \PR, \simi )$, as a function of $\enc{\vec p^*} + \enc{D} + \enc{\Gamma}$, is in $\mathsf{C}$, then we say that $\SHAP(\Q,\PR,\simi)$ is in $\mathsf{C}$ in \emph{data complexity}.\qedhere
    \end{itemize}
\end{defi}

That is, we discuss the computational complexity in terms of variants of the classical combined and data complexity \cite{vardi1982complexity}. In both cases, the similarity function $\simi$ is fixed, as it is part of the problem specification. If not indicated otherwise, our statements about the complexity of\/ $\SHAP(\Q,\PR,\simi)$ always refer to \emph{combined} complexity.
 
\section{General Insights for Fully Factorized Distributions}\label{sec:ind-general}

We first discuss fully factorized distributions. We start by showing that for tractable similarity functions, tractability of $\SHAP( \Q, \IND, \simi )$ in terms of data complexity is guided by the data complexity of $\Q$. This contrasts other query evaluation settings involving probabilities, like query evaluation in probabilistic databases \cite{suciu2011probabilistic,van_den_broeck2017query}. The intuitive reason is that here, the probability distributions are tied to the parameterized query instead of the database. 

\begin{prop}\label{pro:datacomplexity}
    Let $\simi$ be a tractable similarity function and let $\pquery{Q}{\vec R}{\vec P}$ be a \emph{fixed} parameterized query such that $Q_{\vec p}(D)$ can be computed in polynomial time in $\enc{D} + \enc{\vec p}$ for all $\vec p \in \Tup[\vec P]$. Then $\SHAP( \set{Q}, \IND, \simi )$ can be solved in polynomial time; that is, for every class $\Q$ of such queries, $\SHAP(\Q,\IND,\simi)$ can be solved in polynomial time in data complexity.
\end{prop}

\begin{proof}
    Let $\pquery{Q}{\vec R}{\vec P}$ be a fixed parameterized query, and let $(\vec p^*, D, \Gamma)$ be an input to $\SHAP( \set{Q}, \IND, \simi )$, where $\simi$ is a tractable similarity function. As $Q$ is fixed, the number $\ell$ of parameters in $Q$ is constant.
    This means that 
    \begin{enumerate}[(a)]
        \item\label{itm:smallPi} the probability spaces $\Pi_i$, $i \in [\ell]$ (see \Cref{eq:subsetdist}), are of constant size, and
        \item\label{itm:fewparams} $\card{ \supp( \Gamma ) }$ is at most polynomial in $\enc{\Gamma}$ (and can be computed in polynomial time in $\enc{\Gamma}$).
    \end{enumerate}
    By \ref{itm:fewparams}, and since $\simi$ is tractable, we can compute 
    \[
        \E_{ \vec p \sim \Gamma }\bigl[ \simi( \vec p, \vec p^* ) \bigm\vert \vec p_J = \vec p^*_J \bigr] =
        \sum_{ \vec p \in \supp(\Gamma ) } \simi( \vec p, \vec p^* ) \cdot
        \Pr_{ \vec p' \sim \Gamma }\bigl( \vec p' = \vec p \bigm\vert \vec p'_J = \vec p^*_J \bigr)
    \]
    in polynomial time in $\enc{ \Gamma } + \enc{ \vec p^* } + \enc{ D }$ for all $J \subseteq [\ell]$ by brute force. By \ref{itm:smallPi}, computing $\Shap(i)$ by evaluating \eqref{eq:shapleyv2exp} takes, at most, time polynomial in $\enc{ \Gamma } + \enc{ \vec p^* }+ \enc{ D }$.
\end{proof}

For this reason, in the remainder of the paper, we will focus on the \emph{combined complexity} of the computation of $\SHAP(\Q,\IND,\simi)$. 

Next, we point out two relationships with other computational problems. First, we see that $\SHAP( \Q, \IND, \simi )$ is at least as hard as deciding whether the output of a parameterized query is non-empty. To be precise, let $\Q$ be a class of parameterized queries and let $\NONEMPTY(\Q)$ be the problem that takes inputs $( Q, \vec p^*, D)$, and asks to decide whether $\pqueryinst{Q}{\vec p^*}(D) \neq \emptyset$. In order to establish a relationship to the $\SHAP( \Q, \IND, \simi )$ problem, we need two more ingredients.

\begin{defi}\label{def:stronglydependent}
    A similarity function $\simi$ is \emph{left-sensitive}, if for all possible output schemas $\vec R$, one of the following is satisfied:
    \begin{enumerate}[({s}1)]
        \item\makeatletter\def\@currentlabel{s1}\makeatother\label{itm:strongi} for all non-empty $T \in \Rel[\vec R]$ we have $\simi_{\vec R}( \emptyset, \emptyset ) \neq \simi_{\vec R}( T, \emptyset )$; or
        \item\makeatletter\def\@currentlabel{s2}\makeatother\label{itm:strongii} for all non-empty $T \in \Rel[\vec R]$ we have $\simi_{\vec R}( \emptyset, T ) \neq \simi_{\vec R}( T, T )$.
        \qedhere
    \end{enumerate}
\end{defi}

Most natural similarity functions satisfy these, or at least one of these conditions. For example, all the similarity functions from \autoref{exa:similarityfunctions}\ref{itm:setbasedliterature}--\ref{itm:setbasedasym} satisfy (\ref{itm:strongii}), whereas a version of $\NegDiff$ that swaps the roles of $Q$ and $Q^*$ would always satisfy (\ref{itm:strongi}).

As a special case, we will also consider classes of Boolean parameterized queries. If $\vec R$ is a Boolean relation schema, then the only possible inputs to a similarity function $\simi_{\vec R}$ are pairs of $\T$ and $\F$. In particular, $\simi_{\vec R}$ is always tractable. In addition, the properties from \autoref{def:stronglydependent} become
\begin{enumerate}[({s}1{$'$})]
    \item\makeatletter\def\@currentlabel{s1$'$}\makeatother\label{itm:boolstrongi} $\simi_{\vec R}( \F,\F ) \neq \simi_{\vec R}( \T,\F )$ for (\ref{itm:strongi});
    \item\makeatletter\def\@currentlabel{s2$'$}\makeatother\label{itm:boolstrongii} $\simi_{\vec R}( \F,\T ) \neq \simi_{\vec R}( \T,\T )$ for (\ref{itm:strongii}).
\end{enumerate}
This means for Boolean relational schemas $\vec R$ that $\simi_{\vec R}$ is left-sensitive if and only if it is dependent on its first argument: there are $x_1, x_2$ and $y$ such that $\simi_{\vec R}(x_1,y) \neq \simi_{\vec R}(x_2,y)$. This condition is \emph{necessary} for a meaningful notion of $\Shap$ scores: if it does not hold, then $\simi_{\vec R}$ is determined by its second argument and hence the $\Shap$ score w.r.t. $\simi_{\vec R}$ is always $0$.

\begin{thm}\label{thm:shap_hard_if_query_hard}
    Let $\simi$ be left-sensitive. Let $\Q$ be a class of parameterized queries such that for all $\pquery{Q}{\vec R}{\vec P} \in \Q$ there exists $i_0 \in [\ell]$ (which we can find in polynomial time) such that, for all $D \in \DB[\inschema(Q)]$ and all $\vec p = ( p_1,\dots, p_\ell ) \in \Tup[\vec P]$ with \mbox{$p_{i_0} \notin \adom(D)$,} we have $\pqueryinst{Q}{\vec p}(D) = \emptyset$.
    Then $\NONEMPTY(\Q) \cookleq \SHAP( \Q, \IND, \simi )$.
\end{thm}

The restriction imposed on $\Q$ in \cref{thm:shap_hard_if_query_hard} expresses that for every $Q \in \Q$, one of the parameters of $Q$ immediately renders the query result on any $D$ empty, if it is chosen outside the active domain of $D$, and that this parameter can easily be identified.
This is fulfilled, for example, for unions of conjunctive queries for which there exists a parameter that appears in each of its conjunctive queries, and can be extended to Datalog queries for which there exists a parameter that appears in the body of every rule.

\begin{proof}
    Let $\Q$ and $\simi$ be given as in the statement. Let $( \pquery{Q}{\vec R}{\vec P}, \vec p^*, D)$ be an input of $\NONEMPTY(\Q)$ with $\vec p^* = ( p_1^*, \dots, p_\ell^* )$. Without loss of generality, let $i_0 = 1$. First, assume that $\simi = \simi_{\vec R}$ satisfies property (\ref{itm:strongi}).
    
    Fix some $a \notin \adom(D) \cup \set{ p_1^* }$ and let  $\vec p^a = ( a, p_2^*, p_2^*, \dots , p_\ell^* )$. Then $\pqueryinst{Q}{ \vec p^a }(D) = \emptyset$. Let $\Gamma$ be the parameter distribution over $\Tup[\vec P]$ with $\Gamma( \vec p^* ) = \Gamma( \vec p^a ) = \frac12$. Then $\Gamma \in \IND$, because it factorizes into the marginal distributions $(\Gamma_j)_{j \in [\ell]}$ where $\Gamma_1( p_1^* ) = \Gamma_1( a ) = \frac12$ and $\Gamma_2( p_2^* ) = \dots = \Gamma_{\ell}( p_\ell^* ) = 1$. 
    
    We consider the $\SHAP(\Q,\IND,\simi)$ instance $(Q,\vec p^a,D,\Gamma)$. Let $J \subseteq [\ell]$ be arbitrary.
    If $1 \in J$, then $\Pr_{ \vec p \sim \Gamma }( \vec p = \vec p^a \bigm\vert \vec p_J = \vec p_J^a ) = 1$ and, hence,
    \[
        \E_{ \vec p \sim \Gamma }[ \simi( \vec p, \vec p^a ) \bigm\vert \vec p_J = \vec p_J^a ] =
        \simi( \pqueryinst{Q}{\vec p^a}(D), \pqueryinst{Q}{\vec p^a}(D) ) = \simi( \emptyset, \emptyset )\text.
    \]
    If $1 \notin J$, then $\Pr_{ \vec p \sim \Gamma }( \vec p = \vec p^a \bigm\vert \vec p_J = \vec p_J^a )
    = \Pr_{ \vec p \sim \Gamma }( \vec p = \vec p^* \bigm\vert \vec p_J = \vec p_J^a )
    = \frac12$, and, therefore,
    \[
        \E_{ \vec p \sim \Gamma }\bigl[ \simi( \vec p, \vec p^a ) \bigm\vert \vec p_J = \vec p_J^a \bigr] =
        \tfrac12 \simi\bigl( \pqueryinst{Q}{\vec p^*}(D), \emptyset \bigr) + \tfrac12 \simi( \emptyset,\emptyset )\text.
    \]
    Thus, by \eqref{eq:shapleyv2exp}, we have $\Shap(1) = \frac12( \simi( \emptyset, \emptyset ) - \simi( Q_{\vec p^*}(D), \emptyset ))$. Since $\simi$ satisfies (\ref{itm:strongi}), we get
    \begin{equation}\label{eq:shapnonempty}
        \Shap(1) = 0 \Leftrightarrow Q_{\vec p^*}(D) = \emptyset\text.
    \end{equation}
    That is, we can decide $Q_{\vec p^*}(D) = \emptyset$, by solving the $\SHAP(\Q,\IND,\simi)$ instance $(Q,\vec p^a,D,\Gamma)$.

    The proof is analogous for the case that $\simi = \simi_{\vec R}$ satisfies condition (\ref{itm:strongii}) instead of (\ref{itm:strongi}). In this case, we use the $\SHAP(\Q,\IND,\simi)$ instance $(Q, \vec p^*, D, \Gamma )$ with the same distribution $\Gamma$ as above. This leads to $\Shap(1) = \frac12( \simi(\vec p^*,\vec p^*) - \simi(\vec p^a,\vec p^*) )$ and, thus, by (\ref{itm:strongii}), to \eqref{eq:shapnonempty}.
\end{proof}

Van den Broeck et al.~\cite{van_den_broeck2022tractability} have shown a reduction from  computing $\Shap$ to computing the expected similarity. The latter problem, $\ESIM( \Q, \PR, \simi )$, takes the same parameters and inputs as $\SHAP( \Q, \PR, \simi )$, but asks for $\E_{ \vec p \sim \Gamma }[ \simi( \vec p, \vec p^* ) ]$ instead of $\Shap(i)$.

\begin{thmC}[{\cite[Theorem 2]{van_den_broeck2022tractability}}]\label{thm:vdb}
    For any class of parameterized queries $\Q$ and any similarity function $\simi$ such that $\simi\after\Q$ is tractable, $\SHAP(\Q,\IND,\simi) \cookequiv \ESIM( \Q,\IND, \simi )$.
\end{thmC}

\begin{rem}
    The paper \cite{van_den_broeck2022tractability} discusses a more general setting, in which $\Shap$ scores are computed for tractable functions $F$ over $\ell$-ary domains.
    They provide an algorithm that computes $\Shap$ scores of $F$ using oracle calls to the expected value of $F$ with different parameter distributions $\Gamma\in\IND$. This algorithm runs in polynomial time in $\enc{\Gamma}$ and is independent of the choice of $F$ (apart from the oracle, of course).
    For our version of the theorem, we use this algorithm on functions $F_{Q,\vec p^*,D}(\vec p) = \simi(\pqueryinst{Q}{\vec p}(D),\pqueryinst{Q}{\vec p^*}(D))$.
\end{rem}

The key property we needed for \autoref{pro:datacomplexity} was the manageable size of the parameter distributions. This can be formulated as a property of a class of queries.

\begin{defi}\label{def:fewparameters}
    Let $\pquery{Q}{\vec R}{\vec P}$ be a parameterized query and $D \in \DB[\inschema(Q)]$, and denote 
    $
        \parsupp(Q,D) \coloneqq \set{ \vec p \in \Tup[\vec P] \with \pqueryinst{Q}{\vec p}(D) \neq \emptyset }
    $.
    A class $\Q$ of parameterized queries has \emph{polynomially computable parameter support} if there exists a polynomial time algorithm (in $\enc{Q} + \enc{D}$) which, on input $\pquery{Q}{\vec R}{\vec P} \in \Q$ and $D \in \DB[\inschema(Q)]$, outputs a set $P(Q,D) \supseteq \parsupp(Q,D)$.
\end{defi}

There are simple, yet relevant classes with this property, e.g., parameterized CQs with a bounded number of joins, or, more generally, parameterized CQs whose parameters only appear in a bounded number $k$ of atoms: matching these atoms with facts of $D$ in all possible ways yields a set $P(Q,D) \supseteq \parsupp(Q,D)$ with $\card{P(Q,D)} \leq \card{D}^k$. This property allows the efficient computation of expected similarities.

\begin{prop}\label{pro:fewparameters}
    Let $\Q$ have polynomially computable parameter support, and suppose $\simi \after \Q$ is tractable. Then $\SHAP(\Q,\IND,\simi)$ can be solved in polynomial time in combined complexity.
\end{prop}

\begin{proof}
    Let $(Q,\vec p^*,D,\Gamma)$ be an input to $\SHAP(\Q,\IND,\simi)$. 
    As $\Q$ has polynomially computable parameter support, we can efficiently compute a set $P = P(Q,D) \supseteq \parsupp( Q, D )$.
    Then
    \begin{equation*}
        \E_{ \vec p \sim \Gamma }\bigl[ \simi( \pqueryinst{Q}{\vec p}(D), \pqueryinst{Q}{\vec p^*}(D) ) \bigr] =
            \sum_{ \vec p \in P } \Gamma( \vec p ) \cdot \simi( \pqueryinst{Q}{\vec p}(D), \pqueryinst{Q}{\vec p^*}(D) ) + 
            \bigl(1 - \Gamma( P ) \bigr) \cdot \simi( \emptyset, \pqueryinst{Q}{\vec p^*}(D) )\text.
    \end{equation*}
    This value can be computed in polynomial time: As $\simi \after \Q$ is tractable, all function values of $\simi$ in this formula can be computed efficiently. Moreover, since $P$ is of polynomial size, $\Gamma( P ) = \sum_{ \vec p \in P } \Gamma(\vec p)$ can be computed efficiently. Thus, $\ESIM(\Q,\IND,\simi)$ can be solved in polynomial time. This transfers to $\SHAP(\Q,\IND,\simi)$ by \autoref{thm:vdb}.
\end{proof}

\section{Parameterized Conjunctive Queries with Independent Parameters}\label{sec:ind-cqs}

The goal of this section is to identify concrete classes of queries and similarity functions on which the problem of $\Shap$ score computation becomes tractable. 
\Cref{thm:shap_hard_if_query_hard} suggests that we should restrict our attention to query classes whose emptiness problem is tractable in combined complexity; that is, when the query is part of the input.
A natural class of such queries are subclasses of acyclic conjunctive queries.

Recall that the parameters in a parameterized query take an intermediate role between being treated as variables (we allow the parameters to take different values) and constants (in the queries we are really interested in, each parameter will be replaced by a fixed parameter value). This makes the definition of acyclicity ambiguous, as it is not clear whether the parameters should appear as vertices in the query hypergraph or not. However, a simple complexity-theoretic argument shows that the structure induced by the parameters should not be ignored.

\begin{prop}\label{prop:param-important-hypergraph}
    Consider the class $\Qclique$ of parameterized conjunctive queries of the form 
	\[
		Q_\ell(;y_1, \ldots, y_\ell) = \bigwedge_{i=1}^{\ell} V(y_i) \wedge \bigwedge_{1 \leq i < j \leq \ell} E(y_i, y_j)
	\]
	with $\ell \in \mathbb{N}$ and let $\simi$ be left-sensitive. Then, $\SHAP(\Qclique, \IND, \simi)$ is $\sharpP$-hard in combined complexity. 
\end{prop}
\begin{proof}
    First, we observe that $\Qclique$ is tractable since $\pqueryinst{Q}{\vec p}$ only contains constants. Thus, $\simi \after \Qclique$ is tractable, and we show that computing $\ESIM(\Qclique,\IND,\simi)$ is $\sharpP$-hard by reduction from $\sharpClique$ and obtain the desired result by  \autoref{thm:vdb}. The problem $\sharpClique$, given an undirected graph $G=(V_G,E_G)$ and a natural number $k$, asks for the number of cliques of size $k$ in $G$. 

    Now, given $G$ and $k$, we first construct a new graph $G'=(V_{G'},E_{G'})$ as follows: $V_{G'}$ contains all vertices of $V_G$, and $k$ new vertices $v'_1,\dots,v'_k$. If $\simi$ fulfills (\ref{itm:boolstrongi}), then $E_{G'} \coloneqq E_G$. If $\simi$ does not fulfill (\ref{itm:boolstrongi}) but (\ref{itm:boolstrongii}), then $E_{G'} \coloneqq E_G \cup \set{ (v'_i,v'_j) : i \neq j }$.
    In both cases, we can easily deduce the number of $k$-cliques in $G$ from the number of $k$-cliques in $G'$. 

    To define $D$, let $V$ be a unary relation consisting of the vertices in $V_{G'}$, and let $E$ be a binary relation consisting of both orientations of edges in $E_{G'}$. Let $Q = Q_k$,  each $\Gamma_i$ be a uniform distribution over the vertices in $V_{G'}$, and let $\vec p^* = (v'_1, \ldots, v'_k)$. Then $\pqueryinst{Q}{\vec p}(D) = \T$ if and only if the vertices $p_1, \ldots, p_k$ are pairwise distinct and form a clique in $G'$. For each clique in $G'$, each of the $k!$ permutations of the set of clique vertices forms a parameter vector $\vec p$ with $\pqueryinst{Q}{\vec p}(D) = \T$. As all $(n+k)^k$ mappings from $(y_1, \ldots, y_k)$ to $V_{G'}$ have the same probability, in the case of $\simi$ fulfilling (\ref{itm:boolstrongi}), we obtain:
    \[
        \E_{ \vec p \sim \Gamma }\bigl[ \simi( \vec p, \vec p^* ) \bigr] \cdot (n+k)^k = \begin{aligned}[t]
            &\#(k\text{-cliques in }G') \cdot k! \cdot \simi(\T, \F) \\
            &+ \big((n+k)^k - \#(k\text{-cliques in }G') \cdot k!\big)\cdot \simi(\F, \F)\text,
        \end{aligned}
    \]
    which is equivalent to
    \[\#(k\text{-cliques in }G') = \frac{(n+k)^k\big(\E_{ \vec p \sim \Gamma }\bigl[ \simi( \vec p, \vec p^* ) \bigr] - \simi(\F,\F)\big)}{k!\big(\simi(\T,\F) - \simi(\F,\F)\big)}.\]
    For the other case, we only need to swap the second arguments of $\simi$ from $\F$ to $\T$ so in each of the cases, we can efficiently compute the number of $k$-cliques in $G$ from $\E_{ \vec p \sim \Gamma }\bigl[ \simi( \vec p, \vec p^* ) \bigr]$.
\end{proof}

From \autoref{prop:param-important-hypergraph}, we see that the structure induced by the parameters may have a big impact on the computational complexity of the problem. If we choose not to include the parameter structure in our notion of acyclicity of parameterized queries, the query hypergraph of each query in $\Q_{\mathrm{clique}}$ would be empty, and, in particular, acyclic. However, by \autoref{prop:param-important-hypergraph}, computing $\Shap$ scores is hard for this class.
This motivates the following definition. 

\begin{defi}\label{def:p-acyclic}
    We call a parameterized query $Q$ \emph{p-acyclic} (or, a $\pACQ$), if its parameter-including query hypergraph is acyclic. Put differently, $Q$ is p-acyclic, if it is acyclic (in the traditional sense) when parameters are just treated as variables. 
\end{defi}

\begin{exa}\label{exa:acyclic}
    Consider the parameterized CQ
    \[
        Q(x;y_1,y_2) =
        R(x,y_1) \wedge U(y_1,y_2) \wedge V(y_2,x) \text.
    \]

    The queries $Q_{\vec p}$ are acyclic for every choice of parameter values $\vec p$. However, $Q$ is \emph{not} p-acyclic since the (non-parameterized) query
    \[
        Q'(x,x_1,x_2) =
        R(x,x_1) \wedge U(x_1,x_2) \wedge V(x_2,x) 
    \]
    is not acyclic. In contrast, the conjunctive query from \autoref{exa:flights} is p-acyclic.
\end{exa}

If the employed similarity function is based on counting tuples in the query output, a relationship between $\Shap$ score computation for $\pACQ$s and the problem of (weighted) counting of answers to ACQs can be drawn. A central insight for the latter problems is, intuitively, that projection (i.e., the presence of existential quantification) easily renders them intractable, while tractability can be obtained for full ACQs (see \cite{pichler2013tractable} and \cite{durand2014complexity}). In the remainder of this section, we observe a similar phenomenon for the problem of $\Shap$ score computation:
First, we establish $\sharpP$-hardness for a simple class of star-shaped ACQs as in \cite[Theorem 4]{pichler2013tractable} (with non-trivial similarity functions). Next, we contrast this by a tractability result for \emph{full} $\pACQ$s for certain similarity functions based on counts, by leveraging results of \cite[Theorem 1]{pichler2013tractable} and \cite[Corollary 4]{durand2014complexity}.

\begin{restatable}{prop}{cqwithexistshard}\label{pro:nonfull-starACQ-hard}
    Let $\Q$ be the class of Boolean parameterized queries of the form
    \begin{equation}\label{eq:CQ-with-exists-hard}
        \pquery{Q}{}{y_1,\dots,y_{\ell}} = 
        \exists x \with R_1( x, y_1 ) \wedge \dots \wedge R_{\ell}( x, y_{\ell} )
    \end{equation}
    with $\ell \in \mathbb{N}$ and let $\simi$ be left-sensitive. Then $\SHAP( \Q, \IND, \simi )$ is $\sharpP$-hard in combined complexity.
\end{restatable}

\begin{proof}
    Let $\Q$ and $\simi$ be given as described in the statement. Then $\Q$ is tractable: to answer $\pqueryinst{Q}{\vec p}$ on a database $D \in \DB[\inschema(Q)]$, it suffices to scan the relations of $D$ from $R_1$ to $R_\ell$ while keeping track of a set of good values for $x$. However, $\Q$ does \emph{not} have polynomially computable parameter support already if every $y_j$ can take at least two possible values. From the tractability of $\Q$, it follows that $\simi \after \Q$ is tractable.
    
    We show that computing $\ESIM(\Q,\IND,\simi)$ is $\sharpP$-hard by reduction from $\sharpPDNF$ which, given a positive formula in disjunctive normal form, asks for the number of satisfying assignments. This problem is $\sharpP$-hard (\cite{DBLP:journals/siamcomp/ProvanB83}) and we then obtain the desired result from \autoref{thm:vdb}.
    
    Let $\phi$ be a positive DNF formula in propositional variables $X_1,\dots,X_{\ell}$ with disjuncts $\phi_1,\dots,\phi_n$.
    Define
    \[
        D = \set[\big]{ R_i(j,1) \with i \in [\ell] \text{ and } j \in [n] } \cup \set[\big]{ R_i(j,0) \with i \in [\ell] \text{ and } X_i \text{ does not appear in } \phi_j }\text,\\
    \]
    and let $Q$ be the parameterized query from \eqref{eq:CQ-with-exists-hard}.
    Moreover, let $\Gamma_i(0) = \Gamma_i(1) = \frac12$ for all $i \in [\ell]$, and let $\Gamma \in \IND$ be the product distribution of $\Gamma_1, \dots, \Gamma_{\ell}$.
    
    Every $\vec p \in \set{0,1}^{\ell}$ corresponds to a truth assignment $\alpha_{\vec p}$ of $X_1,\dots,X_{\ell}$ via $\alpha_{\vec p}(X_i) = \T$ if $p_i = 1$, and $\alpha_{\vec p}(X_i) = 0$ if $p_i = 0$, for all $i \in [\ell]$. 
    
    We claim $\pqueryinst{Q}{\vec p}(D) = \T$ if and only if $\alpha_{\vec p}( \phi ) = \T$. If this claim holds, we are done:
    Suppose that $\simi$ satisfies (\ref{itm:boolstrongii}), i.e., $\simi(\F,\T) \neq \simi(\T,\T)$, and let $\vec p^* = (1,\dots,1)$. (If $\simi$ only satisfies (\ref{itm:boolstrongi}), then the proof is done similarly with $\vec p^*=(0, \dots,0)$.) We obtain    
    \[
    \E_{ \vec p \sim \Gamma }\bigl[ \simi( \vec p, \vec p^* ) \bigr] =
    \frac{\#\phi}{2^{\ell}} \cdot \simi(\T,\T) + \Bigl(1-\frac{\#\phi}{2^{\ell}}\Bigr) \cdot \simi(\F,\T)\text,
    \]
    
    where $\#\phi$ denotes the number of satisfying assignments of $\phi$. Since $\simi(\F,\T) \neq \simi(\T,\T)$, $\#\phi$ can be computed efficiently from $\E_{\vec p \sim \Gamma}\bigl[ \simi(\vec p, \vec p^* ) \bigr]$.

    It only remains to prove our claim: $\pqueryinst{Q}{\vec p}(D) = \T$ if and only if $\alpha_{\vec p}( \phi ) = \T$, for all $\vec p = (p_1,\dots,p_{\ell}) \in \set{0,1}^{\ell}$. Fix an arbitrary $\vec p \in \set{0,1}^{\ell}$.

    If $\pqueryinst{Q}{\vec p}(D) = \T$, there exists $j \in [n]$ such that $D \supseteq \set{ R_1( j, p_1 ), R_2( j, p_2 ), \dots, R_{\ell}( j, p_{\ell} ) }$. For those $i$ with $p_i = 0$, by definition of $D$, the variable $X_i$ does not appear in $\phi_j$. Hence, $p_i = 1$ for every variable $X_i$ that appears in $\phi_j$. Therefore, $\alpha_{\vec p}( \phi_j ) = \T$, so $\alpha_{\vec p}$ satisfies $\phi$.

    For the other direction, let $\alpha_{\vec p}(\phi) = \T$. Then $\alpha_{\vec p}$ satisfies a disjunct, say, $\phi_j$. By definition, $D$ contains $R_i(j,p_i)$ for all $i \in [\ell]$ with $p_i = 1$. If $p_i = 0$, then $X_i$ cannot appear in $\phi_j$ (since $\alpha_{\vec p}(\phi_j) = \T$) and, hence, $D$ contains $R_i(j,p_i)$. Together, $\pqueryinst{Q}{\vec p}(D) = \T$.
\end{proof}

\begin{rem}
    This, and our remaining hardness results, are stated for the class $\IND$ of fully factorized distributions. Inspection of our proofs, however, reveals that these hardness results hold even for the much more restricted subclass of $\IND$ where every parameter takes one of two possible values with probability $\frac12$.
\end{rem}

Next, we show a tractability result for parameterized \emph{full} ACQs. In contrast, by \autoref{pro:nonfull-starACQ-hard}, even a single existential quantifier can make the problem difficult. A similar observation has been made for the (weighted) counting for ACQ answers \cite{pichler2013tractable,durand2014complexity} (which we use to establish tractability here).

\begin{restatable}{prop}{fullACQsEasyCases}\label{pro:full-ACQs-IntDiffCDiff-easy}
    Let $\Q$ be the class of full $\pACQ$s, and let $\simi$ be any of\/ $\Intersection$, $\NegSymDiff$, or $\NegDiff$. Then $\SHAP(\Q,\IND,\simi)$ can be solved in polynomial time in combined complexity.
\end{restatable}

\begin{rem}
    The previous proposition is stated for specific similarity functions $\simi$ because it does not hold for arbitrary similarity functions. For example, computing $\Shap$ scores for $\simifont{MinDiff}_{A,B}$ and $\simifont{ExpMinDiff}_{A,B}$ is $\sharpP$-hard  in combined complexity on full $\pACQ$s as they may project out variables, so we mimic the proof of \autoref{pro:nonfull-starACQ-hard} to obtain $\sharpP$-hardness on queries of the form $Q(x_1, x_2, x_3) = S(x_1,x_2) \wedge R_1(x_3, y_1) \wedge \ldots \wedge R_\ell(x_3,y_\ell)$ where $Q.A$ and $Q.B$ are the attributes corresponding to $x_1$ and $x_2$, respectively.
\end{rem}

Before we come to the proof of \autoref{pro:full-ACQs-IntDiffCDiff-easy}, let us do some preparation. We introduce an artificial ``similarity'' function $\Count$ with $\Count(T_1,T_2) = \card{T_1}$ for all $T_1,T_2 \in \Rel[\vec R]$. 
Then for all parameterized queries $\pquery{Q}{\vec R}{\vec P}$, all $\vec p \in \Tup[\vec P]$, and $D \in \DB[\inschema(Q)]$ we have
\begin{align*}
    &\NegSymDiff( \vec p, \vec p^* ) =
    2 \cdot 
    \Intersection( \vec p, \vec p^* ) 
    - \Count( \vec p, \vec p^* ) 
    - \card{ \pqueryinst{Q}{\vec p^*}(D) }\text,\\
    &\NegDiff( \vec p, \vec p^* ) =
    \Intersection( \vec p, \vec p^* ) - \card{ \pqueryinst{Q}{\vec p^*}(D) }\text,
\end{align*}
where again $\simi(\vec p,\vec p^*)$ is shorthand for $\simi( \pqueryinst{Q}{\vec p}(D), \pqueryinst{Q}{\vec p^*}(D) )$.
The value $-\card{ \pqueryinst{Q}{\vec p^*}(D) }$ is a constant additive term (depending neither on $\vec p$ nor on $J$) in every utility value $\nu(J)$ for $\NegSymDiff$, and therefore cancels out in the Shapley value. For all $i \in [\ell]$, linearity of expectation yields
\begin{align*}
    &\Shap_{\NegSymDiff}(i) =
    2 \cdot \Shap_{\Intersection}(i) - \Shap_{\Count}(i)\text,\\
    &\Shap_{\NegDiff}(i) =
    \Shap_{\Intersection}(i)\text.
\end{align*}
Hence, the tractability for $\NegSymDiff$ and $\NegDiff$ follows, once we have established the tractability for both $\Count$ and $\Intersection$.

In the proof of \autoref{pro:full-ACQs-IntDiffCDiff-easy} we will use a reduction to the weighted answer counting problem for ACQs, denoted by $\sharpwACQ$: 
A $\sharpwACQ$ instance is a triple $(Q',D',w)$, consisting of a (non-parameterized) ACQ $Q'$, a database $D'$, and a weight function $w$ on the domain elements. The problem consists in computing the sum of the values $w(a_1,\dots,a_n) = \prod_{ i \in [n] } w(a_i)$\text, for all $(a_1,\dots,a_n)$ that fulfill $D' \models Q'( a_1,\dots, a_n)$.

\begin{proof}[Proof of {\autoref{pro:full-ACQs-IntDiffCDiff-easy}}]
    We first show tractability of $\SHAP(\Q,\IND,\Count)$. By \cite[Theorem 1]{pichler2013tractable}, ${\Count}\after\Q$ is tractable. Hence,  
    \[
        \SHAP(\Q,\IND,\Count) \cookequiv \ESIM(\Q,\IND,\Count)
    \]
    by \autoref{thm:vdb}, and it suffices to show that $\ESIM(\Q,\IND,\Count)$ can be computed in polynomial time.

    We first show that from an $\ESIM(\Q,\IND,\Count)$ instance $(Q,\vec p^*,D,\Gamma)$, we can construct, in polynomial time, a $\sharpwACQ$ instance $(Q',D',w)$ such that
    \begin{equation}\label{eq:wACQ-reduction}
        \E_{ \vec p \sim \Gamma }\bigl[ \card{ \pqueryinst{Q}{\vec p}(D) } \bigr] =
        \sum_{ \vec a \with D' \models Q'(\vec a)} w( \vec a )\text,
    \end{equation}
    where $Q'$ is a full ACQ. Since, by \cite[Corollary 4]{durand2014complexity}, $\sharpwACQ$ can be solved in polynomial time for the class of full ACQs, this claim entails tractability of $\ESIM(\Q,\IND,\Count)$.

    So let $(Q,\vec p^*,D,\Gamma)$ be an instance of $\ESIM(\Q,\IND,\Count)$, where $Q = \pquery{Q}{\vec x}{\vec y}$ with parameter variables $\vec y = (y_1,\dots,y_{\ell})$. We now define an input $(Q',D',w)$ of $\sharpwACQ$ as follows:
    \begin{itemize}
        \item For every $i \in [\ell]$, let $R_i$ be a new binary relation symbol, and let $y_i'$ be a new variable. Define    
            \[
                Q'(\vec x, \vec y, \vec y') = Q \wedge R_1(y_1,y_1') \wedge \dots \wedge R_{\ell}(y_{\ell},y'_{\ell})\text.
            \]
            This is a normal relational query with variables $\vec x$, $\vec y$ and $\vec y'$. Since $Q$ is a full $\pACQ$, $Q'$ is a full ACQ (the hypergraph of $Q'$ differs from that of $Q$ by having $\ell$ additional, independent ears).
        \item For all $i \in [\ell]$ let $A_i$ be the set of active domain elements $a$ of $D$ such that there exists an $R$-atom in $Q$, and an $R$-tuple $t$ in $D$, such that $a$ appears in $t$ in the same position as $y_i$ appears in the $R$-atom of $Q$. For every $a \in A_i$, we introduce a new domain element $a^{(i)}$, and let 
            \[
                D' = D \cup \set[\big]{ R_i( a, a^{(i)} ) \with a \in A_i, i\in[\ell] }\text.
            \] 
        \item We define a weight function $w$ on domain elements by letting $w( a^{(i)} ) = \Gamma_i( a )$ for all $a \in A_i$ and $i \in [\ell]$. We let $w( a ) = 1$ for all other $a$.
    \end{itemize}
    For our claim, it remains to show that \eqref{eq:wACQ-reduction} holds.

    The output schema of $Q'$ is $\outschema(Q') = \vec R \times \vec P \times \vec P'$. The elements of $\Tup[ \vec P' ]$ are tuples of the shape $(p_1^{(1)},\dots,p_{\ell}^{(\ell)})$ for $(p_1,\dots,p_{\ell}) \in \Tup[\vec P]$. 
    We let $\alpha$ be the function that maps $(p_1,\dots,p_{\ell}) \in A_1 \times \dots \times A_{\ell}$ to $(p_1^{(1)},\dots,p_{\ell}^{(\ell)})$.
    Let
    \[
        S = \set[\big]{ (\vec c, \vec p, \vec p') \in \Tup[ \outschema(Q') ] \with D' \models Q'( \vec c, \vec p, \vec p' ) }\text.
    \]
    Then from the construction of $Q'$ and $D'$, it follows that for all $(\vec c, \vec p, \vec p') \in S$, we have $\vec p \in A_1 \times \dots \times A_{\ell}$ and $\vec p' = \alpha( \vec p )$.
    Further, for all $\vec p \in \Tup[\vec P]$, we define
    \[
        S_{ \vec p } = \set[\big]{ \vec c \in \Tup[\vec R] \with ( \vec c, \vec p, \alpha(\vec p) ) \in S }\text,
    \]
    and, moreover, let $P \coloneqq \set{ \vec p \in \Tup[ \vec P ] \with S_{\vec p} \neq \emptyset }$.

    Generally, for all $\vec c \in \Tup[ \vec R ]$ and all $\vec p \in P$, we have 
    \[
        D \models Q_{\vec p}( \vec c ) \Leftrightarrow
        D' \models Q'\bigl( \vec c, \vec p, \alpha( \vec p ) \bigr) \Leftrightarrow \vec c \in S_{\vec p}\text.
    \]
    (If $\vec p \notin P$, then both $\pqueryinst{Q}{\vec p}(D)$ and $S_{\vec p}$ are empty.)
    We obtain
    \[
        \card{ S_{\vec p} } = \card{ \set{ \vec c \in \Tup[\vec R] \with D \models \pqueryinst{Q}{\vec p}(\vec c) } } = \card{ \pqueryinst{Q}{\vec p}(D) }\text.
    \]
    Now, we calculate
    \begin{align*}
        \sum_{ (\vec c, \vec p, \vec p') \in S } w( \vec c, \vec p, \vec p' ) &=
        \sum_{ \vec p \in P } w( \alpha(\vec p) ) \sum_{ \vec c \in S_{\vec p} } \underbrace{w( \vec p ) \cdot w(\vec c)}_{=1} =
        \sum_{ \vec p \in P } \Gamma( \vec p ) \cdot \card{ S_{\vec p} } \\ 
        &=
        \sum_{ \vec p \in P } \Gamma( \vec p ) \cdot \card{ \pqueryinst{Q}{\vec p}(D) }
        = \E_{\vec p \sim \Gamma }\bigl[ \card{ \pqueryinst{Q}{\vec p}(D) } \bigr]\text.
    \end{align*}
    As per our discussion from before, this establishes tractability for $\SHAP(\Q,\IND,\Count)$.

    For $\Intersection$, we reduce to $\SHAP(\Q,\IND,\Count)$ as follows. For a given full $\pACQ$ $Q$, and reference parameter $\vec p^*$, we define
    \[
        \pquery{Q^{\cap}}{\vec x}{\vec y} \coloneqq 
        \pquery{Q}{\vec x}{\vec y} \wedge \pqueryinst{Q}{\vec p^*}(\vec x)
        \text.
    \]
    Since $Q$ is a full $\pACQ$, so is $Q^{\cap}$: The hypergraph for $Q^{\cap}$ contains all the hyperedges of the hypergraph of $Q$, and from the second conjunct we additionally get subhyperedges of hyperedges of $Q$. In the hypergraph of $Q^{\cap}$, the latter are thus immediately ears and can be removed by the GYO algorithm \cite{graham,DBLP:conf/compsac/YuO79}. Since $Q$ is a full $\pACQ$, this entails that so is $Q^{\cap}$.
    Moreover, we have
    \[
        \Intersection( \pqueryinst{Q}{\vec p}(D), \pqueryinst{Q}{\vec p^*}(D) ) =
        \card[\big]{ \pqueryinst{Q}{\vec p}(D) \cap \pqueryinst{Q}{\vec p^*}(D) } =
        \card[\big]{ \pqueryinst{Q^{\cap}}{\vec p}(D) } =
        \Count( \pqueryinst{Q^{\cap}}{\vec p}(D), \pqueryinst{Q^{\cap}}{\vec p^*}(D) )
        \text.
    \]
    Thus, the $\Shap$ scores for $(Q,\vec p^*,D,\Gamma)$ using $\Intersection$ coincide with those of $(Q^{\cap}, \vec p^*, D, \Gamma )$ using $\Count$. The remaining claims follow, since the associated $\Shap$ scores are multilinear in those of $\Intersection$ and $\Count$.
\end{proof}

The tractability in \autoref{pro:full-ACQs-IntDiffCDiff-easy} for $\Intersection$ may seem surprising, given that \cite[Proposition~5]{durand2014complexity} states that counting the answers in the intersection of two full ACQs is $\sharpP$-complete. However, the queries we intersect in our proof are of a very special shape, as they only concern different parameterizations 
\emph{of the same parameterized query}, one of them being fixed. From this we get that our intersection query is acyclic too. When two \emph{arbitrary} acyclic queries are intersected like this, acyclicity can be lost.

\section{Extending Conjunctive Queries with Filters}\label{sec:filters}

In this section we investigate parameterized conjunctive queries $Q$ which contain conjuncts of the shape $f(\vec x, \vec y)$, called \emph{filters}, where for every valuation $(\vec a,\vec b)$ of $(\vec x,\vec y)$, we have $f(\vec a,\vec b) \in \set{\T,\F}$. This value may depend on the input database $D$ on which $Q$ is evaluated. 

\begin{defi}\label{def:pcqfilters}
    A \emph{parameterized CQ with filters} is an expression of the shape
    \begin{equation}\label{eq:pcqfilters}
        Q(\vec x_f; \vec y) = \exists \vec x_b \colon
        \alpha_1( \vec x_1, \vec y_1 ) \wedge \dots \wedge \alpha_n( \vec x_n, \vec y_n ) 
        \wedge
        f_1( \vec x_1', \vec y_1') \wedge \dots \wedge f_m( \vec x_m', \vec y_m' ) \text,
    \end{equation}
    where
    \begin{itemize}
        \item $\vec x_f$ and $\vec x_b$ are all variables in the query (bound or free), hence, all $\vec x_i$ and all $\vec x_j'$ consists of variables appearing among $\vec x_f$ or $\vec x_b$;
        \item $\vec y$ is the tuple of all parameters in the query, hence, all $\vec y_i$ and all $\vec y_j'$ consist of parameters from $\vec y$;
        \item each $\alpha_i(\vec x_i, \vec y_i)$ is a relational atom containing the variables $\vec x_i$ and the parameters in $\vec y_i$ (and no other variables or parameters); and
        \item every variable appearing in some filter; that is, in $\vec x'_1 \cup \dots \cup \vec x'_m$, also occurs in some relational atom; that is, in $\vec x_1 \cup \dots \cup \vec x_n$.\qedhere
    \end{itemize}
\end{defi}

We call $Q$ a \emph{p-acyclic CQ (or $\pACQ$) with filters}, if its query hypergraph (including the parameters) is acyclic, when $f_i( \vec x_i', \vec y_i' )$ is interpreted as a relation $F_i( \vec x_i', \vec y_i' )$. Moreover, we say that a class $\mathcal Q$ of parameterized CQs with filters has \emph{uniformly tractable filters}, if there exist a constant $r$ and a polynomial $\phi$ such that
\begin{itemize}
    \item the arity of its filters is bounded by $r$, and
    \item the function $(\vec x', \vec y', D) \mapsto f(\vec x',\vec y')$ can be computed in time $\phi\bigl(\enc{\vec x'} + \enc{\vec y'} + \enc{D}\bigr)$ for every filter $f$ that appears among queries in $\Q$.
\end{itemize}

Some examples of filters are shown in \autoref{tab:filters}. For example, $f(x,y) = \T$ if and only if $x \leq y$, is a filter. In practice, we would often just directly write $x \leq y$ in our query expressions, and write $f = [x \leq y]$ for the filter. Filters such as $[x \leq y]$, where $y$ is a parameter, typically denote infinite relations. If filters are uniformly tractable, they can be materialized as finite relations (\cref{pro:full-PACQs-easy}), but in general, they can make computing $\Shap$ scores $\sharpP$-hard (\cref{pro:nonfull-starACQ-ineq-hard}).

\begin{table}
\caption{Further examples of filters}\label{tab:filters}
\begin{tabular}{ll}\toprule
    \textbf{Filter concept} & \textbf{Examples}\\\midrule
    (Linear) (in)equalities & $f(x) = [x = c]$\\
                            & $f(x,y) = [x \leq y]$\\
                            & $f(x_1,x_2,y) = [x_1 + 2x_2 \geq y]$\\\midrule
    Other comparisons       & $f(x,y) = \T$ if and only if $x$ and $y$ are \\
                            & strings with the same first character \\\midrule
    Views                   & $f(x_1,x_2) = \T$ if and only if the input\\
                            & satisfies $\neg\exists z\colon R(x_1,z) \wedge R(z,x_2)$\\
    \bottomrule
\end{tabular}
\end{table}

\begin{exa}\label{exa:flightsfilters}
    Let us revisit our flights example (\autoref{exa:flights}). 
    Consider the following parameterized query with filters, 
    $\pquery{Q}{ x_1,x_2,t_{\mathsf{arr}} }{ d, t_{\mathsf{dep}}, L_{\min}, L_{\max}, T }$: 
    \begin{equation*}
        \begin{aligned}
        \exists h,a_1,a_2,t_1,t_2\colon
        {}&\mathsf{Flights}( x_1, d, a_1, \mathtt{CDG}, h, t_{\mathsf{dep}}, t_1 )
        \wedge \mathsf{Flights}( x_2, d, a_2, h, \mathtt{JFK}, t_2, t_{\mathsf{arr}} )\\
        {}&\wedge (t_1 + L_{\mathsf{min}} < t_2)
        \wedge (t_2 < t_1 + L_{\mathsf{max}})
        \wedge (t_{\mathsf{arr}} \leq T)\text.
        \end{aligned}
    \end{equation*}
    This query is asking for 2-hop flights between $\mathtt{CDG}$ and $\mathtt{JFK}$. It is parameterized with the date $d$, a desired departure time $t_{\mathsf{dep}}$, a desired layover time interval $( L_{\mathsf{min}}, L_{\mathsf{max}} )$ and a latest local arrival time $T$. It returns the flight IDs $x_1$, $x_2$, and the final arrival time $t_{\mathsf{arr}}$.

    The query has three filters making use of inequalities:
    \begin{align*}
        f_1 &= [t_1 + L_{\mathsf{min}} < t_2]\text, &
        f_2 &= [t_2 < t_1 + L_{\mathsf{max}}]\text, &
        f_3 &= [t_{\mathsf{arr}} \leq T]\text.
    \end{align*}
    To highlight that the concept of filters is rather general, imagine another filter stipulating that the airlines $a_1$ and $a_2$ are from the same country. That is,
    \[
        f(a_1,a_2) = \T 
        \quad\text{if and only if}\quad
        \exists c\colon \mathsf{Airline}(a_1,c) \wedge \mathsf{Airline}(a_2,c)\text.
        \qedhere
    \]
\end{exa}

\subsection{Parameter Importance in Conjunctive Queries with Filters}\label{sec:parameter-shap-scores-with-filters}

Our tractability result from the previous section can be transferred to computing $\Shap$ scores for full $\pACQ$s with uniformly tractable filters. 
The idea is simple: We just have to materialize the relations represented by the filters, and apply \cref{pro:full-ACQs-IntDiffCDiff-easy}.

\begin{prop}\label{pro:full-PACQs-easy}
    Let $\mathcal Q$ be a class of full $\pACQ$s with uniformly tractable filters, and let $\simi$ be any of $\Intersection$, $\NegSymDiff$, or\/ $\NegDiff$. Then $\SHAP(\Q,\IND,\simi)$ can be computed in polynomial time  in combined complexity.
\end{prop}

\begin{rem}\label{rem:count-also-works}
    Of course, this result relies on \cref{pro:full-ACQs-IntDiffCDiff-easy} and, in principle, also holds for any similarity function to which \cref{pro:full-ACQs-IntDiffCDiff-easy} can be extended. In particular, it extends to the artificial similarity function $\Count$ which was introduced for the proof of \cref{pro:full-ACQs-IntDiffCDiff-easy}. This will prove useful in the second part of this section.
\end{rem}

\begin{proof}
    Let $(Q,\vec p^*,D,\Gamma)$ be an input to $\SHAP(\Q,\IND,\simi)$ where $\simi$ is one of $\Intersection$, $\NegSymDiff$, or $\NegDiff$. Suppose $Q = \pquery{Q}{\vec x}{\vec y}$ with $\vec x = (x_1,\dots,x_n)$ and $\vec y = (y_1,\dots,y_{\ell})$.

    We construct a $\pACQ$ $Q'$ without filters by replacing every $f_i( \vec x_i', \vec y_i' )$ with a relational atom $F_i( \vec x_i', \vec y_i' )$, where $F_i$ is a new relation symbol. Then, by definition, $Q'$ is a full $\pACQ$.

    Next, we construct a new database instance $D'$ which materializes the relevant tuples which pass every filter:
    \[
        D' \coloneqq 
        D \cup \bigcup_{i = 1}^m \bigl\{ F_i(\vec a, \vec b) : 
        \begin{aligned}[t]
            &f_i(\vec a,\vec b) = \T \text{ and}\\
            &\vec a \text{ is a valuation of } \vec x_i' \text{ over the active domain of $D$ and}\\
            &\vec b \text{ is a valuation of } \vec y_i' \text{ over } {\textstyle\bigcup_{j \in [\ell] } \supp( \Gamma_j ) }\bigr\}
        \end{aligned}
    \]
    The database $D'$ can be computed in time polynomial in $\enc{Q} + \enc{D} + \enc{\Gamma}$ because $\Q$ has uniformly tractable filters: By the bounded arity of the filters, there are only polynomially many tuples $\vec a$ and $\vec b$ to consider in the construction of $D'$, and we can evaluate $f_i(\vec a, \vec b) = \T$ for every such tuple in polynomial time as well.

    Let $\vec p = ( p_1, \dots, p_{\ell} ) \in \Tup[\vec P]$ have positive probability in $\Gamma$. Recall that if $f_i( \vec x_i', \vec y_i' )$ is a filter in $Q$, then all variables of $\vec x_i'$ appear among the relational atoms of $Q$. Hence, if $\vec a \in \pqueryinst{Q}{\vec p}(D)$, then $\vec a$ consists of constants appearing in $D$. Because $D'$ coincides with $D$ on the relations which are present in $Q$, this also holds for all $\vec a \in \pqueryinst{Q'}{\vec p}(D')$.

    Moreover, by construction, for any valuation $\vec a$ of $\vec x_i'$ over the active domain of $D$, we have that
    $f_i( \vec a, \vec p_i' ) = \T$ if and only if
    $F_i( \vec a, \vec p_i' )$ holds in $D'$. Hence, $\pqueryinst{Q}{\vec p}(D) = \pqueryinst{Q'}{\vec p}(D')$.
    Therefore, $(Q,\vec p^*,D,\Gamma)$ and $(Q',\vec p^*, D',\Gamma)$ yield the same $\Shap$ scores. Since $Q'$ is a full $\pACQ$, the claim follows from \cref{pro:full-ACQs-IntDiffCDiff-easy}.
\end{proof}

Again, projection can easily complicate computation, so the requirement that queries be full cannot be dropped from \cref{pro:full-PACQs-easy}. This already happens for the simplest single-atom threshold queries, which ask whether some tuple satisfies given upper bounds as in a search form, whereas single-atom queries without filters are tractable by \cref{pro:fewparameters}.

\begin{prop}\label{pro:nonfull-starACQ-ineq-hard}
    Let $\Q$ be the class of $\pACQ$s with filters of the shape
    \begin{equation}\label{eq:BCQineq_ell}
        \pquery{Q}{}{y_1,\dots,y_{\ell}} = 
        \exists x_1,\dots,x_{\ell} \with R(x_1, \dots, x_{\ell}) \wedge (x_1 \leq y_1) \wedge \dotsm \wedge (x_{\ell} \leq y_{\ell})\text,
    \end{equation}
    with $\ell \in \mathbb{N}$, and let $\simi$ be left-sensitive.
    Then $\SHAP(\Q,\IND,\simi)$ is $\sharpP$-hard  in combined complexity.
\end{prop}

\begin{proof}
    The class $\Q$ above is tractable: To answer $\pqueryinst{Q}{\vec p}$ on $D$, it suffices to check, for every $R(a_1,\dots,a_{\ell})$ in $D$, whether $a_i \leq p_i$ holds for all $i \in [\ell]$. However, $\Q$ does not have polynomially computable parameter support. Tractability of $\Q$ implies tractability of $\simi \after \Q$, and we can use \cref{thm:vdb} again to discuss $\ESIM(\Q,\IND,\simi)$ instead of $\SHAP(\Q,\IND,\simi)$.

    Similar to the proof of \cref{pro:nonfull-starACQ-hard}, we show hardness by reduction from $\sharpPDNF$. 
    Let $\phi$ be a positive DNF formula in propositional variables $X_1,\dots,X_{\ell}$ with disjuncts $\phi_1, \dots, \phi_n$.
    For all $i \in [n]$, let $\vec a^i = (a^i_1, \dots, a^i_{\ell}) \in \set{0,1}^{\ell}$ with
    \begin{equation}\label{eq:aijdef}
        a^i_j = 
        \begin{cases}
            1   & \text{if } X_j \text{ appears in } \phi_i \text,\\
            0   & \text{otherwise,}
        \end{cases}
    \end{equation}
    for all $j \in [\ell]$, and define $D = \set{ R( \vec a^1 ), \dots, R( \vec a^{n} ) }$. Let $\Gamma_j(0) = \Gamma_j(1) = \frac12$ for all $j \in [\ell]$, and let $\Gamma \in \IND$ be the product distribution of $\Gamma_1, \dots, \Gamma_{\ell}$. Moreover, let $Q$ be the parameterized query from \eqref{eq:BCQineq_ell}.
    With every $\vec p \in \set{0,1}^{\ell}$ we associate a truth assignment $\alpha_{\vec p}$ that sets $X_i$ to $\T$ if and only if $p_i = 1$, for all $i \in [\ell]$. 

    We claim that, then, for all $\vec p \in \set{0,1}^{\ell}$, we have $\pqueryinst{Q}{\vec p}(D) = \T$ if and only if $\alpha_{\vec p}( \phi ) = \T$. Once this is established, we are done, since we can proceed exactly like in the proof of \cref{pro:nonfull-starACQ-hard}. We therefore now just prove the claim.

    First, suppose that $\alpha_{\vec p}$ is a satisfying assignment for $\phi$. Then $\alpha_{\vec p}( \phi_i ) = 1$ for some $i \in [\ell]$. In particular, $p_j = 1$ whenever $X_j$ appears in $\phi_i$. By construction, $a^i_j \leq p_j$ for all $j \in [\ell]$. As $R( \vec a^i ) \in D$, we have $\pqueryinst{Q}{\vec p}(D) = \T$.

    Second, suppose $\pqueryinst{Q}{\vec p}(D) = \T$. Then one of the $R(\vec a^i)$ in $D$ satisfies $a^i_j \leq p_j$ for all $j \in [\ell]$. The positions $j$ for which $a^i_j = 1$ belong to the variables of $\phi_i$ and, from the inequalities, we have $p_j = 1$ for these positions $j$ as well. This entails $\alpha_{\vec p}( \phi_i ) = 1$, i.e., $\alpha_{\vec p}$ satisfies $\phi$.
\end{proof}

\begin{rem}\label{rem:continuous}
    In practice, inequality-based filters would typically be used for attributes with continuous domains, like $\RR$. Continuous parameter values would require a treatment of continuous parameter distributions. Our framework can handle continuous parameter distributions as follows.
    
    On the technical side, the definition of the $\Shap$ score needs to be changed to avoid conditional probabilities. This is because for continuous parameter distributions, the events $\vec p_J = \vec p_J^*$ which appear as conditions in the definition of the $\Shap$ score always have probability $0$ for $J \neq \emptyset$. However, for fully factorized distributions, this issue can be easily resolved, since the distribution $\vec p_{[\ell]\setminus J}$ is given by the factorization, independently of $\vec p_J$.
    
    On the algorithmic side, we can discretize the continuous distribution into intervals defined by the values in the database to obtain a discrete distribution with finite support that yields the same output. This construction works for inequality filters such as in \eqref{eq:BCQineq_ell}, and many others, but, of course, not for all types of filters.
\end{rem}

\subsection{Application: Why-Not Questions}
\label{sec:why-not}

In this section, we consider the application of our framework to the problem of explaining \emph{why-not questions}. A why-not question specifies a tuple $\vec t$ absent from the result set of a query $Q$, despite being expected to appear. A prominent paradigm for answering such questions is the \emph{operator-based} (or \emph{query-based}) model, where the explanation consists of query operators that contributed to the absence of the tuple~\cite{ChapmanJ09,BidoitHT14,BidoitHT15}. 

In the case of conjunctive queries (with projection), an \e{explanation} $E$ is a collection of Boolean operators from the query (what we later call ``filters'') that invalidate all tuples $\vec u$ whose projection is $\vec t$, and, moreover, the removal of $E$ leads to the survival of at least one such $\vec u$ (hence $\vec t$). Note that many explanations $E$ may exist. Prior work studied the generation of the \emph{latest} explanation based on the query structure~\cite{ChapmanJ09} or a representation (via a polynomial) of the space of \emph{all} explanations~\cite{BidoitHT14,BidoitHT15}. The former is concise but can be arbitrary and overly uninformative, and the latter gives the general picture but can be overwhelming due to many explanations. We discuss a third and complementary approach: assign to each operator an attribution score based on its contribution to the absence of $\vec t$. This gives a perspective that is general and yet concise. As usual, we adopt the Shapley value as the scoring model for the contribution of an operator to a collective operation.

Since the players are now the filters of a query rather than its parameters, we take a step back and first consider conjunctive queries with filters \emph{without any kind of parameterization} (we connect the resulting Shapley values to $\Shap$ scores of parameters in \cref{sec:why-not-connection}).
The queries we consider are shaped as in \cref{def:pcqfilters} but without parameters (i.e., the parameter tuple $\vec y$ from \eqref{eq:pcqfilters} is empty), that is,
\begin{equation}\label{eq:whynot-filteredcq}
   Q(\vec x_f) = \exists \vec x_b\colon
    \alpha_1( \vec x_1 ) \wedge \dots \wedge \alpha_n( \vec x_n ) \wedge 
    f_1( \vec x_1' ) \wedge \dots \wedge f_m( \vec x_m' ) \text, 
\end{equation}
requiring that every variable appearing in some filter; that is, in $\vec x'_1 \cup \dots \cup \vec x'_m$, also occurs in some relational atom; that is, in $\vec x_1 \cup \dots \cup \vec x_n$.

Given a database $D$ and a conjunctive query $Q$ with filters as above, any tuple $\vec t$ over the output schema of $Q$ but with $\vec t \notin Q(D)$ raises the \emph{why-not question} \enquote{why does tuple $\vec t$ not appear in $Q(D)$?}.

\begin{exa}\label{exa:flights-why-not}
    Consider again the query $Q$ of \autoref{exa:flightsfilters}, but assume that the parameters ($d, t_{\mathsf{dep}}, L_{\mathsf{min}}, L_{\mathsf{max}}$, $T$) have been turned into constants ($\mathtt{d}, \mathtt{t}_{\mathsf{dep}}, \mathtt{L}_{\mathsf{min}}, \mathtt{L}_{\mathsf{max}}, \mathtt{T}$), i.e.,   
    \begin{equation*}
        Q( x_1,x_2,t_{\mathsf{arr}} )=
        \begin{aligned}[t]
        \exists h,a_1,a_2,t_1,t_2 \colon
        {}&\mathsf{Flights}( x_1, \mathtt{d}, a_1, \mathtt{CDG}, h, \mathtt{t}_{\mathsf{dep}}, t_1 )\\
        {}&\wedge \mathsf{Flights}( x_2, \mathtt{d}, a_2, h, \mathtt{JFK}, t_2, t_{\mathsf{arr}} )\\
        {}&\wedge (t_1 + \mathtt{L}_{\mathsf{min}} < t_2)
        \wedge (t_2 < t_1 + \mathtt{L}_{\mathsf{max}})
        \wedge (t_{\mathsf{arr}} \leq \mathtt{T})\text.
        \end{aligned}
    \end{equation*}
    Now let $\vec t$ be a why-not question, where we require that $\vec t \notin Q(D)$.
    The \emph{intermediate tuples} are, in this case, the tuples in the self-join of $\mathsf{Flights}$ on $h$ that match the constants (join conditions and constants are not filters), and those \emph{relevant for $\vec t$} agree with $\vec t$ on $x_1, x_2, t_{\mathsf{arr}}$, that is, they form $Q_\emptyset[\vec t](D)$, cf.~\eqref{eq:Q_J}. The \enquote{operators} causing such intermediate tuples to disappear are our filters $f_1 = [t_1 + L_{\mathsf{min}} < t_2]$ and $f_2 = [t_2 < t_1 + L_{\mathsf{max}}]$ and $f_3 = [t_{\mathsf{arr}} \leq T]$.

    This way, the set of intermediate tuples with respect to $\vec t$, and the set of operators is uniquely determined. Rewriting our query, however, allows us to modify these notions: The instantiation of variables with constants can also cause tuples to disappear, and we could also start from a \emph{Cartesian product} of $\mathsf{Flights}$ with itself to define the notion of intermediate tuple. To analyze these effects as well, we can choose to include them explicitly as filters. Moreover, instead of two separate filters $[t_1 + \mathtt{L}_{\mathsf{min}} < t_2]$ and $[t_2 < t_1 + \mathtt{L}_{\mathsf{max}}]$, we could have a single filter $[ t_1 + \mathtt{L}_{\mathsf{min}} < t_2 < t_1 + \mathtt{L}_{\mathsf{max}}]$. 

    The filter $f \in \set{f_1,f_2,f_3}$ contributes to the absence of $\vec t$ in the query output, if there is an intermediate tuple relevant for $\vec t$ which does not satisfy the filter condition of $f$. 

    For example, any intermediate tuple relevant for $\vec t$, in which we have $t_1 > t_2$, is removed by $f_1$ (but not by $f_2$). Some of these intermediate tuples may, additionally, also be removed by $f_3$.
\end{exa}

This highlights that multiple filters can affect the presence of $\vec t$, and they can do so in different ways. In general, it can happen that every filter only removes some of the intermediate tuples projecting to a given tuple $\vec t$, but the collection of all filters removes all of them.
This interpretation of a \emph{joint contribution} of a set of filters to the absence of $\vec t$ naturally suggests modeling filters as players in a cooperative game. Subject to a suitable utility function, the Shapley values of the filters then reveal their impact on the absence of $\vec t$ in the query output. Answering why-not questions then becomes a problem of computing Shapley values in this model. Next, we introduce and analyze this problem for two natural utility functions.

\subsubsection{A Cooperative Game for Why-Not Questions}

Let $Q$ be a CQ with filters, let $D$ be a database over $\inschema(Q)$ and let $\vec t$ be a tuple over $\outschema(Q)$ such that $\vec t \notin Q(D)$. 
A cooperative game for the filter contribution consists of the set of players $I_Q = [m]$, so $\set{ f_i : i \in I_Q }$ are the filters appearing in $Q$, together with a utility function $\nu_{Q,D,\vec t}$ assigning a real number to each subset of $I_Q$. If $Q$, $D$, and $\vec t$ are clear from the context, we omit the subscripts and just write $(I,\nu)$ instead of $(I_Q, \nu_{Q,D,\vec t})$.

The following two utility functions are natural:
\begin{enumerate}
    \item a qualitative utility function $\vqual$, where $\vqual(J) \in \set{0,1}$ indicates whether $\vec t$ is removed due to the filters in $J$ alone (as in set semantics); and
    \item a quantitative utility function $\vquan$, where $\vquan(J)$ reports how many tuples projecting to $\vec t$ are removed by the filters in $J$, that is, by how much they reduce the multiplicity of $\vec t$ under bag semantics.
\end{enumerate}

To define $\vqual$ and $\vquan$ formally, we let $Q_J$ denote the full CQ with filters which is obtained from $Q$ by removing all filters $f_i$ with $i \notin J$, and by removing all existential quantifiers. That is, if $Q$ is of the shape \eqref{eq:whynot-filteredcq}, then
\begin{equation}\label{eq:Q_J}    
    Q_J(\vec x) = \alpha_1( \vec x_1 ) \wedge \dots \wedge \alpha_n( \vec x_n ) \wedge \bigwedge_{i \in J}
    f_i( \vec x_i' )
    \text.
\end{equation}

Moreover, we let $Q_J[\vec t]$ denote the query obtained from $Q_J$ by replacing the variables which were free in $Q$ by their valuation according to $\vec t$. Note that $Q_J[\vec t]$ still has variables if $Q$ contained existential quantifiers.

Our utility functions can now be formally defined as
\begin{equation}\label{eq:vqual}
    \vqual(J) \coloneqq 
    \begin{cases}
        1 & \text{if } \big\lvert Q_J[\vec t](D) \big\rvert = 0\\
        0 & \text{otherwise,}
    \end{cases}
\end{equation}
and
\begin{equation}\label{eq:vquan}
    \vquan(J) \coloneqq \bigl\lvert Q_{\emptyset}[\vec t](D) \setminus Q_J[\vec t](D) \bigr\rvert.
\end{equation}

For a why-not question $\vec t$, the Shapley values for the associated cooperative games then map each filter to its contribution to the absence of $\vec t$ in the query output. Hence, we get two computational problems, one using $\vqual$, and one using $\vquan$. 

\begin{prob}
    Let $\mathcal Q$ be a class of CQs with filters. We have two computational problems,
    \begin{enumerate}
        \item $\WhyNotShapleyQual(\mathcal Q)$: Given $(Q,D,\vec t)$, compute $\Shapley(I,\vqual,i)$ for all $i \in I$; and
        \item $\WhyNotShapleyQuan(\mathcal Q)$: Given $(Q,D,\vec t)$, compute $\Shapley(I,\vquan,i)$ for all $i \in I$;
    \end{enumerate}
    where $Q \in \mathcal Q$ is a CQ with filters $\set{ f_j : j \in I }$, $D$ is a database over the input schema of $Q$, $\vec t$ is a tuple over the output schema of $Q$ with $\vec t \notin Q(D)$, and $\vqual$ and $\vquan$ are defined from $Q$, $D$, and $\vec t$ as in \eqref{eq:vqual} and \eqref{eq:vquan}.
\end{prob}

As for $\SHAP$, we measure the complexity of these problems as a function of $\enc{Q} + \enc{D} + \enc{\vec t}$ (combined complexity) or, for a fixed $Q$, as a function of $\enc{D} + \enc{\vec t}$ (data complexity). In data complexity, both problems can be solved in polynomial time whenever the filters can be evaluated in polynomial time, since the number of filters, and thus of coalitions, is constant.

\begin{rem}
    In the literature, $\vec t$ does not necessarily need to be a full tuple in the output schema, but rather a tuple pattern that may contain undefined values~\cite{BidoitHT15}. This can also be modelled in our framework by adding existential quantifiers to $Q$ for the undefined values in $\vec t$. By that, all of our results naturally extend to tuple patterns in the output schema.
\end{rem}

\begin{exa}
    Let us pick up a shortened version of \autoref{exa:flights-why-not} where we already materialized the self-join of the $\mathsf{Flights}$-relation and, for better readability, project out the attributes $x_1, x_2, h, a_1, a_2$ to obtain a relation $\mathsf{TwoHopConnections}(t_\mathsf{arr}, t_1, t_2)$, i.e.,
    \[
       \mathsf{TwoHopConnections}(t_\mathsf{arr},t_1,t_2) = 
       \begin{aligned}[t]
       \exists x_1,x_2,h,a_1,a_2\colon
       {}&\mathsf{Flights}( x_1, \mathtt{d}, a_1, \mathtt{CDG}, h, \mathtt{t}_{\mathsf{dep}}, t_1 )\\
       {}&\wedge \mathsf{Flights}( x_2, \mathtt{d}, a_2, h, \mathtt{JFK}, t_2, t_{\mathsf{arr}} )
       \end{aligned}
    \]
    For simplicity, assume that timestamp attributes take integer values and are ordered naturally, and assume that the constants $L_{\mathsf{min}}$, $L_{\mathsf{max}}$, and $T$ are given as $1$, $4$, and $8$, respectively. The query, thus, becomes
    \[
        Q(t_\mathsf{arr}) = 
        \exists t_1, t_2\colon \mathsf{TwoHopConnections}(t_\mathsf{arr}, t_1, t_2) \wedge (t_1 + 1 < t_2) \wedge (t_2 < t_1 + 4) \wedge (t_{\mathsf{arr}} \leq 8)
        \text,
    \]
    with the three filters, $f_1 = [ t_1 + 1 < t_2 ]$, $f_2 = [ t_2 < t_1 + 4 ]$, and $f_3 = [ t_{\mathsf{arr}} \leq 8 ]$.
    
    Suppose our input database $D$ (the materialization of $\mathsf{TwoHopConnections}$) consists of the three tuples $(7,1,5), (7,1,2)$ and $(7,2,6)$, as shown in \autoref{tab:2hopdb}. Then the output of $Q$ on $D$ is empty, and we can ask \enquote{Why $7 \notin Q(D)$?}

    \begin{table}[ht]
    \caption{An input database $D$}\label{tab:2hopdb}\small
    \begin{tabular}{ccc}\toprule
        \textbf{Arrival} ($t_{\mathsf{arr}}$) & \textbf{HopArrival} ($t_1$) & \textbf{HopDeparture} ($t_2$)\\\midrule
        $7$ & $1$ & $5$\\
        $7$ & $1$ & $2$\\
        $7$ & $2$ & $6$\\\bottomrule
    \end{tabular}
    \end{table}
    
    We observe that $f_1$ removes the second tuple in $\mathsf{TwoHopConnections}$, $f_2$ removes the first and third tuple and $f_3$ does not remove any tuple. Hence, the Shapley values of the filters with respect to $\vqual$ are $\frac{1}{2}$, $\frac{1}{2}$, and $0$. In contrast, for $\vquan$, the Shapley values are $1, 2$ and $0$. Thus, under the qualitative utility function $\vqual$, $f_1$ and $f_2$ are deemed equally important to the absence of $7$, whereas under the quantitative utility function $\vquan$, $f_2$ is deemed \emph{more} important than $f_1$.
\end{exa}
To get a feeling for the behavior of $\vqual$ and $\vquan$, we first consider the highly restricted class of CQs with filters which contain only a single relational atom. Even on this simple class, they behave quite differently.

\begin{thm}\label{thm:vqual-hard}
    Let $\Q$ be the class of CQs with filters of shape
    \[
        Q() = \exists x_1, \dots, x_n\colon 
        R(x_1,\dots,x_n) \wedge \bigwedge_{ i \in I } x_i'  = c\text,
    \]
    where $c$ is a constant. Then $\WhyNotShapleyQual(\Q)$ is $\sharpP$-hard  in combined complexity.
\end{thm}
Note that for Boolean CQs with filters, the only possible why-not question that can be asked is \enquote{why $\emptytuple \notin Q(D)$?}, where $\emptytuple = ()$ is the empty tuple. This is equivalent to asking why the query returns $\F$ on $D$.  
\begin{proof}
    We show this by giving a Turing reduction from the $\sharpSetCover$ problem: On input a number $m$ and a family $\mathcal{S} = (S_1, \ldots, S_n)$ of subsets of $[m]$, $\sharpSetCover$ asks for the number of ways to cover $[m]$ with a subfamily of the subsets, i.e., $\card[\big]{\set{J \subseteq [n] \, \mid \, \bigcup_{i \in J} S_i = [m]}}$. The $\sharpP$-hardness of this problem follows from \cite{DBLP:journals/siamcomp/ProvanB83}.

    For a given instance $(m, \mathcal{S})$ of $\sharpSetCover$, we construct a filtered query $Q$ and a single relation $R$ such that
    \begin{enumerate} 
        \item $R = \set{ \vec t^j : j \in [m]}$, where $\vec t^j$ encodes which sets $S_i$ cover $j$; and such that
        \item each filter $f_i$ removes exactly those $\vec t^j$ with $j \in S_i$. 
    \end{enumerate}
    To do that, let $R$ be of arity $n$ and let  $\vec t^j$ be defined by 
    \[
        \vec{t}^{j}_i \coloneqq 
        \begin{cases}
        1 & \text{if } j \in S_i\text,\\
        0 & \text{otherwise.}
        \end{cases}
    \] 
    That is, the $i$th attribute column in $R$ represents the indicator function of the set $S_i \subseteq [m]$.
    We define
    \[
        Q() = \exists x_1 \ldots x_n \colon R(x_1,\dots,x_n) \wedge (x_1 = 0) \wedge \dots \wedge (x_n = 0)\text,
    \] 
    and let $f_i = [x_i = 0]$ for $i = 1, \dots, n$.
    With this setup, for any subset $J \subseteq [n]$ of filters, we have $\vqual(J) = 1$ if and only if the corresponding $(S_i)_{i \in J}$ cover $[m]$, and counting these subsets of filters solves the $\sharpSetCover$ instance.

    In general, however, a single oracle call to $\WhyNotShapleyQual$ will not be sufficient to count these subsets. Instead, we will define $n$ different instances $(m_1, \mathcal{S}_1),\dots,(m_n, \mathcal{S}_n)$ of $\sharpSetCover$, call the oracle on each of the corresponding relations and filtered queries, and use these $n$ values combined to solve the original $\sharpSetCover$ instance $(m, \mathcal{S})$. 
    To this end, for $\ell$ in $[n]$, we define 
    \begin{align*}
        m_\ell &= m + \ell\text,&
        \mathcal S_\ell &=  \set[\big]{S_1, \ldots, S_n, \set{m + 1}, \ldots, \set{m + \ell}}\text.
    \end{align*}
    To cover $[m_{\ell}] = [m + \ell]$ with sets from $\mathcal S_\ell$ we need a set cover of $[m]$ from $\mathcal S$, since the numbers $\leq m$ are not contained in the added sets, and all the sets $\set{m + j}$ for $1\leq j \leq \ell$, since each of them contains a number that is contained in no other set. 
    Just as we defined $R$ and $Q$ from $(m,\mathcal S)$, we define a relation $R_{\ell}$ and a query $Q_{\ell}$ for every instance $(m_{\ell}, \mathcal S_{\ell})$, $\ell \in [n]$ in the exact same way.

    Let $(I_{\ell},\vqual_{\ell})$ be the cooperative game for the empty why-not question $\emptytuple$ with respect to $Q_{\ell}$ on $R_{\ell}$ using the qualitative utility function. In this game, the Shapley value for filter $f_{n+1}$ is given as follows: 
    In any permutation $\sigma$ of the filters, the filter $f_{n+1}$ makes a difference if and only if it completes a set cover of $[m+\ell]$, that is, if and only if the filters preceding $f_{n+1}$ in $\sigma$ consist of $f_{n+2},\dots,f_{n+\ell}$ together with filters $f_j$, $j \in [n]$, whose sets $S_j$ cover $[m]$. If these form a set cover of size $k$, then there are $k + \ell - 1$ filters before $f_{n+1}$ in $\sigma$ and $n-k$ after it. Thus, we have
     \begin{equation*}
        w(\ell) \coloneqq \Shapley( I_{\ell}, \vqual_{\ell}, n+1 ) =
        \sum_{k=1}^n \frac{(k + \ell - 1)!(n-k)!}{(n + \ell)!} \cdot s_k
        \text,
    \end{equation*}
    where  
    \[
        s_k \coloneqq \bigl\lvert\set{J\subseteq [n] \, \mid \, \card{J} = k \text{ and } \bigcup_{i \in J} S_i = [m]}\bigr\rvert
    \]
    is the number of set covers of $[m]$ with exactly $k$ sets.
    This yields the following system of linear equations:
	\[
	\begin{pmatrix}
		1!(n-1)! &     2!(n-2)! & \ldots &      n!0! \\
		2!(n-1)! &     3!(n-2)! & \ldots &  (n+1)!0! \\
		\vdots   &     \vdots   & \ddots &  \vdots  \\
		n!(n-1)! & (n+1)!(n-2)! & \ldots & (2n -1)!0!
	\end{pmatrix} 
	\cdot 
	\begin{pmatrix}
		s_1 \\
		s_2 \\
		\vdots \\
		s_n
	\end{pmatrix} 
	= 
	\begin{pmatrix}
	(n+1)!w(1) \\
	(n+2)!w(2)\\
	\vdots \\
	(2n)!w(n)
	\end{pmatrix}.
	\]
    The values $w(\ell)$ can be obtained by calling the oracle for $\WhyNotShapleyQual$ on input $(Q_{\ell},R_{\ell},\emptytuple,n+1)$ for $\ell = 1, \dots, n$. 
    The matrix $a_{i,j} = (i+j-1)!(n-j)!$ is invertible, as shown by \cite[proof of Theorem 1.1.]{bacher2002determinants}, so we can determine $s_1,\dots,s_n$ in polynomial time and return $\sum_{k=1}^n s_k$. By definition of the values $s_k$, this is the solution of the $\sharpSetCover$ instance $(m,\mathcal S)$.
\end{proof}

Fortunately, the Shapley value with respect to the quantitative utility function $\vquan$ can be evaluated efficiently, even on much more expressive classes than that of \cref{thm:vqual-hard}. We first prove a simple closed-form expression for solutions to $\WhyNotShapleyQuan$ on a simple class of queries. After that, we will generalize our tractability results to acyclic CQs with filters.

\begin{prop}
    Let $Q() = \exists x_1 \ldots x_n \colon R(x_1,\dots,x_n) \wedge \bigwedge_{i \in I_Q} f_i(\vec x_i')$ and let $R$ be a relation. For each tuple $\vec t \in R$, let $k_{\vec t}$ denote the number of filters $f_i$ that remove $\vec t$. Then for query $Q$, database (relation) $R$, the empty why-not question $\emptytuple$, and filter index $i$, we get 
    \begin{equation}\label{eq:shap-kt}
	\Shapley(I_Q,\vquan_{Q,R,\emptytuple},i) = 
        \sum_{\substack{\vec t \in R\text{, }\\\vec t\text{ filtered out by } f_i}} \frac{1}{k_{\vec t}}.
    \end{equation}
\end{prop}

\begin{proof}
    The idea of the proof is to exploit additivity by rewriting the utility function as a sum of indicator utility functions that are easy to evaluate.
    For $\vec t \in R$, we define the indicator utility function
    \[
    \nu_{\vec t}(J) \coloneqq \begin{cases}
        1 & \text{if } \vec t \in Q_{\emptyset}(D) \setminus Q_J(D)\text{, and}\\
        0 & \text{otherwise,}
        \end{cases}
    \]
    which enables us to rewrite the quantitative utility function as 
    \[
    \vquan_{Q,R,\emptytuple}(J) 
    = 
    \card[\big]{Q_{\emptyset}(D) \setminus Q_J(D) }
    = 
    \sum_{\vec t \in R} \nu_{\vec t}(J).
    \]
    Additivity of the Shapley value yields
    \[
        \Shapley( I_Q, \vquan_{Q,R,\emptytuple}, i ) =
        \sum_{ \vec t \in R } \Shapley( I_Q, \nu_{\vec t}, i )\text.
    \]
    
    Now, we observe that $\nu_{\vec t}(J \cup \set{i}) - \nu_{\vec t}(J) = 1$ if and only if $f_i$ removes the tuple $\vec t$ and none of the filters $f_j$ for $j \in J$ removes $\vec t$. Otherwise, $\nu_{\vec t}(J \cup \set{i}) - \nu_{\vec t}(J) = 0$. Thus, $\Shapley(I_Q,\nu_{\vec t},i)$ is simply the probability of $f_i$ being the first filter removing $\vec t$ in a random permutation of the filters. This is the probability of $f_i$ being picked first among the $k_{\vec t}$ filters removing $\vec t$. Hence, 
    $\Shapley(I_Q,\nu_{\vec t},i) =  \frac{1}{k_{\vec t}}$ if $\vec{t}$ is filtered out by $f_i$, and $0$ otherwise, implying the claim.
\end{proof}

If $\Q$ is a class of queries matching the shape of the query $Q$ in the above proposition, and it has uniformly tractable filters, the identity \eqref{eq:shap-kt} allows us to compute the Shapley values to solve $\WhyNotShapleyQuan(\mathcal Q)$ efficiently. We may even allow the queries in this class to have some free variables (and, hence, more general why-not questions $\vec t \neq \emptytuple$): All we need to do is to plug in the values of $\vec t$ into the free variables in $Q$ and use \eqref{eq:shap-kt}.

\begin{cor}\label{cor:why-not-simple}
    Let $\Q$ be a class of CQs with uniformly tractable filters of the shape
    $Q(\vec x_f) = \exists \vec x_b \colon \alpha(\vec x) \wedge \bigwedge_{ i \in I } f_i( \vec x_i' )$ where $\vec x$ are all variables appearing in the query, and $\vec x_f$ and $\vec x_b$ are the free and bound ones, and $\alpha(\vec x)$ is a relational atom containing all variables. Then  $\WhyNotShapleyQuan(\Q)$ can be computed in polynomial time  in combined complexity.
\end{cor}

In the remainder of this section, we show that Corollary \ref{cor:why-not-simple} can be lifted to acyclic conjunctive queries. Instead of giving an explicit Yannakakis-style algorithm, we want to show how to use our framework for the importance of query parameters in this setting. We do this by first connecting Shapley values and $\Shap$ scores from a very general point of view, and then by representing filtered queries with parameterized queries. That allows us to derive the statement we aim for from our previous results and shows connections between seemingly unrelated questions.

\subsubsection{Shapley values and $\Shap$ scores for binary decisions}

The goal of this section is to show how we can compute Shapley values using an oracle that computes $\Shap$ scores. For that, let us take a step back to the basic setting for the Shapley value, as described in \cref{sec:pre-shapleyvalue}:
Consider the set of players $I = [\ell]$ and a utility function $\nu : 2^{[\ell]} \to \RR$. Using characteristic vectors, we can also see $\nu$ as a function from $\set{0,1}^{\ell}$ to $\RR$, where the $i$th bit indicates whether player $i$ is part of the coalition. This interpretation allows us to introduce the $\Shap$ score in this setting. 

To this end, let $\vec{\pi} = ( \pi_1, \dots, \pi_{\ell} )$ be a vector of probabilities, and assume that the players' memberships in the coalition are independent and such that player $i$ is part of the coalition with probability $\pi_i$. This yields a fully factorized distribution $\Gamma_{\vec\pi}$ on $\set{0,1}^{\ell}$.

We consider two computational problems:
\begin{enumerate}
    \item $\probShapley$: On input $([\ell],\nu)$, compute $\Shapley([\ell],\nu,i)$ for all $i \in [\ell]$.
    \item $\binSHAP$:  On input $([\ell], \nu, \vec\pi)$ with $\pi_j  > 0$ for all $j \in [\ell]$, compute the $\Shap$ scores of all players $i \in [\ell]$ w.r.t.~$\nu$ and $\Gamma_{\vec\pi}$ and reference parameter $\vec 1 = (1,\dots,1)$; that is, compute $\Shapley([\ell], \nu_{\mathsf{SHAP}}[\vec\pi], i)$ for all $i \in [\ell]$, where
    \begin{equation}\label{eq:nushap}
        \nu_{\mathsf{SHAP}}[\vec\pi](J) \coloneqq \E_{\vec x \sim \Gamma_{\vec\pi}}\bigl[\nu(\set{i \,\mid\, x_i = 1}) \, \mid \, \forall i \in J \colon x_i = 1\bigr]
        \text.
    \end{equation} 
\end{enumerate}

At this point, we do not need to worry about the encoding of the utility function in the inputs of the problems. We later specialize both problems in a way that the encoding of $\nu$ is clear, and unproblematic.

The difference between $\nu$ and $\nu_{\mathsf{SHAP}}[\vec\pi]$ in \eqref{eq:nushap} is the role of players not belonging to the coalition $J$: In $\nu$, they are simply not part of the coalition, while in $\nu_{\mathsf{SHAP}}[\vec\pi]$, they randomly join the coalition with probability $\pi_i$. This immediately implies the following observation:

\begin{obs}
    For $\vec 0 = (0, \ldots, 0)$, we have  $\nu_{\mathsf{SHAP}}[\vec 0] = \nu$ and, hence, 
    \[
        \Shapley([\ell], \nu, i) = 
        \Shapley([\ell],\nu_{\mathsf{SHAP}}[\vec 0], i)
        \text.
    \]
\end{obs}

Our goal is to reduce $\probShapley$ to $\binSHAP$.
If $\vec \pi = \vec 0$ were a valid input to the $\binSHAP$ problem, this would now be immediate.
However, under $\Gamma_{\vec 0}$, the reference parameter $\vec 1$ does not have positive probability. We resolve this by using interpolation.

\begin{thm}\label{thm:Shapley-via-SHAP}
    $\Shapley([\ell],\nu,i)$ can be computed with polynomial overhead by using $\ell + 1$ oracle calls to $\binSHAP$ on inputs of the shape $([\ell],\nu,(x,\dots,x))$, where $x > 0$.
\end{thm}

\begin{proof}
    For $x \in (0,1]$,  let $\vec \pi = (x, \ldots, x)$. With this probability vector, the utility function of the $\binSHAP$ problem becomes
    \[
    \nu_{\mathsf{SHAP}}[(x,\ldots,x)](J) = \sum_{J' \subseteq [\ell] \setminus J} \nu(J \cup J') \cdot \prod_{j \in J'} x \cdot \prod_{j \in [\ell] \setminus (J \cup J')} (1-x) \text{,}
    \]
    which is a polynomial in $x$ of degree at most $\ell - \card{J}$.

    Now, consider the map $\phi : x \mapsto \Shapley([\ell], \nu_{\mathsf{SHAP}}[(x,\ldots,x)], i)$. As the Shapley value is a linear combination of the values that the utility function attains, $\phi$ is a polynomial in the variable $x$ of degree at most $\ell$. The constant term of this polynomial is $\phi(0)$, which is, by the previous observation, equal to $\Shapley([\ell], \nu, i)$. We can interpolate this value by evaluating $\phi$ for $\ell + 1$ pairwise different arguments $x_1,\dots,x_{\ell} \in (0,1]$, using oracle calls to $\binSHAP$ on inputs $([\ell],\nu,(x_j,\dots,x_j))$ for $j \in [\ell]$. Multiplying the resulting vector, from the left, with the inverse of the corresponding Vandermonde matrix yields the coefficients of $\phi$, in particular, $\phi(0) = \Shapley([\ell],\nu,i)$.
\end{proof}

\subsubsection{The Connection to \texorpdfstring{$\Shap$}{SHAP} Scores}\label{sec:why-not-connection}

We are finally able to state and prove the main result of this section.

\begin{thm}
    Let $\Q$ be a class of acyclic filtered queries with uniformly tractable filters. Then $\WhyNotShapleyQuan(\Q)$ can be computed in polynomial time  in combined complexity.
\end{thm}

\begin{proof}
    We want to prove this theorem by combining \cref{thm:Shapley-via-SHAP}  and \cref{pro:full-PACQs-easy}. For that, note that $\WhyNotShapleyQuan(\Q)$ is a special version of $\probShapley$ where the utility function $\nu$ is not given explicitly but rather encoded in $Q, D$ and $\vec t$. In the same way, we define a special version $\WhyNotSHAPQuan(\Q)$ of $\binSHAP$. This is the problem that, on input $Q$, $D$, $\vec t$, and 
    $\vec\pi \in (0,1]^I$ (assigning a probability $\pi_i$ to each filter $f_i$ of the filtered query $Q$),
    asks for the $\Shap$ scores of all filters $f_i$ rather than their Shapley values; that is, to compute $\Shapley(I, \vquan_{\mathsf{SHAP}}[\vec \pi], i)$ for all $i\in I$ where
    \[\vquan_{\mathsf{SHAP}} [\vec\pi](J)
    \coloneqq \E_{\vec x \sim \Gamma_{\vec\pi}}\bigl[\vquan(\set{i \, \mid \, x_i = 1}) \, \mid \, \forall i \in J \colon x_i = 1\bigr]
        \text.
    \]
    
    By \cref{thm:Shapley-via-SHAP}, $\WhyNotShapleyQuan(\Q)$ can be solved using a linear number of oracle calls to $\WhyNotSHAPQuan(\Q)$. Now, we are looking for a class of full $\pACQ$ with uniformly tractable filters  $\Q'$ and a similarity function $\simi$ from \cref{pro:full-PACQs-easy} that allow reducing $\WhyNotSHAPQuan(\Q)$ to $\SHAP(\Q', \IND, \simi)$.

    Towards this goal, let $(Q, D, \vec t, \vec \pi)$ be an input to $\WhyNotSHAPQuan(\Q)$ and consider the following construction: From the filtered query $Q \in \Q$ of the form $Q(\vec x_f) = \exists \vec x_b\colon
    \alpha_1( \vec x_1 ) \wedge \dots \wedge \alpha_n( \vec x_n ) \wedge 
    f_1( \vec x_1' ) \wedge \dots \wedge f_m( \vec x_m' )$ and the tuple $\vec t$ over $\outschema(Q)$, we construct $Q'$ as follows. We first extend each filter $f_j(\vec x'_j)$ with a new Boolean parameter $y_j$ to obtain the filter $f'_j(\vec x'_j, y_j)$, which is $f_j(\vec x'_j)$ if $y_j = 1$, and $\T$ if $y_j = 0$. Then, we substitute all free variables $\vec x_f$ with their corresponding values from $\vec t$ and remove the existential quantifiers. Hence,
    \begin{align*}
        Q'(\vec x_b; \vec y) = 
        \begin{aligned}[t]
            &\alpha_1( \vec x_1 )[\vec x_f / \vec t] \wedge \dots \wedge \alpha_n( \vec x_n )[\vec x_f / \vec t]  \\
            &\wedge f_1'( \vec x'_1, y_1 )[\vec x_f / \vec t] \wedge \dots \wedge f'_m( \vec x'_m, y_m)[\vec x_f / \vec t]
            \text.
        \end{aligned}
    \end{align*}
    Then $Q'$ is a full $\pACQ$ and the class $\Q'$ of $Q'$ which are constructed this way from queries of $\Q$ is a class of full $\pACQ$ with uniformly tractable filters. Furthermore, by construction, we obtain for all $J \subseteq I$:
    \begin{equation}\label{eq:Q_J-Q_chi}
        Q_J[\vec t](D) = Q'_{\chi(J)}(D)\text,
    \end{equation}
    where $\chi(J)$ is the characteristic vector of the set $J$. To recall, the left-hand side of \eqref{eq:Q_J-Q_chi} is evaluating the unquantified version of $Q$ which has been restricted to the filters in $J$ and where the free variables have been replaced by $\vec t$, as defined in \eqref{eq:Q_J}. The right-hand side is evaluating our parameterized query $Q'$ with parameter $\chi(J)$ as defined in \eqref{eq:Q_p}.

    We now want to use \eqref{eq:Q_J-Q_chi} to relate $\vquan_{\SHAP}[\vec\pi]$ to a utility function corresponding to $\SHAP(Q', \IND, \simi)$ for some similarity function $\simi$ as in \eqref{eq:def-nu-SHAP}. First, we observe
    \begin{align*}
        \vquan_{Q,D,\vec t}(J) = \card{ Q_{\emptyset}[\vec t](D) \setminus Q_J[\vec t](D) } 
        =
        \card{ Q_{\emptyset}[\vec t](D) } - \card{ Q_J[\vec t](D) }
        =
        \card{ Q_{\emptyset}[\vec t](D) } - \card{ Q'_{\chi(J)}(D) }\text,
    \end{align*}
    where the first second holds, since $Q_{\emptyset}[\vec t](D) \supseteq Q_J[\vec t](D)$. Then, since the value $\card{ Q_{\emptyset}[\vec t](D) }$ does not depend on $J$, we have
    \begin{align}
        \vquan_{\mathsf{SHAP}}[\vec\pi](J) ={}& \E_{\vec x \sim \Gamma_{\vec\pi}}\bigl[\vquan(\set{i \,\mid\, x_i = 1}) \, \mid \, \forall i \in J \colon x_i = 1\bigr]
        \notag{}
        \\
        {}={}&
        \E_{\vec x \sim \Gamma_{\vec\pi}}\bigl[\card{ Q_{\emptyset}[\vec t](D) } - \card{ Q'_{\chi(\set{i \,\mid\, x_i = 1})}(D)} \, \mid \, \forall i \in J \colon x_i = 1\bigr]\notag{}\\
        {}={}&\card{ Q_{\emptyset}[\vec t](D) } - \E_{\vec x \sim \Gamma_{\vec\pi}}\bigl[
        \card{ Q'_{\chi(\set{i \,\mid\, x_i = 1})}(D)} \, \mid \, \forall i \in J \colon x_i = 1\bigr].\label{eq:utilitymagic}
    \end{align}
    Finally, let 
    \begin{align*}
        \nu'(J) {}\coloneqq{}& \E_{\vec x \sim \Gamma_{\vec\pi}}\bigl[
        \card{ Q'_{\chi(\set{i \,\mid\, x_i = 1})}(D)} \, \mid \, \forall i \in J \colon x_i = 1\bigr] \\
        {}={}&\E_{\vec x \sim \Gamma_{\vec\pi}}\bigl[
        \card{ Q'_{\vec x}(D)} \, \mid \, \vec x_J = \vec 1_J\bigr]\text.
    \end{align*} 

    From \eqref{eq:utilitymagic}, we immediately get \[\Shapley([m], \vquan_{\mathsf{SHAP}}[\vec\pi], i) = - \Shapley([m], \nu', i)\text.\]
    
    We also observe that $\nu'(J)$ is the utility function from \eqref{eq:def-nu-SHAP} with $\simi = \Count$ on input $(Q',\vec 1, D, \Gamma_{\vec\pi})$. Hence, we have constructed an instance of $\SHAP(\Q',\IND,\Count)$ whose $\Shap$ scores are equal to the $\Shap$ scores of $\WhyNotSHAPQuan$ up to a sign flip.
    Since $\Q'$ is a class of full $\pACQ$ with uniformly tractable filters, 
    we can apply our results of \cref{sec:parameter-shap-scores-with-filters}: \cref{pro:full-PACQs-easy}, and its extension to the $\Count$ \enquote{similarity}.
    Hence, we can solve our constructed instance $(Q',\vec 1,D,\Gamma_{\vec\pi})$ of $\SHAP(\Q',\IND,\Count)$ in polynomial time, which completes our reduction.
\end{proof}

\section{Correlated Parameters and Approximability}\label{sec:cor-approx}

In this section, we allow classes $\PR$ of parameter distributions with correlations. 
We only need to make the following tractability assumptions (which are trivially satisfied by $\IND$).
\begin{enumerate}
    \item For every fixed $\ell$, and every $\Gamma \in \PR_{\vec P}$ with $\len{ \vec P } = \ell$, the support $\supp( \Gamma )$ can be computed in polynomial time in $\enc{ \Gamma }$. 
    \item For all $\vec p$ and $J$, we can compute $\Pr_{ \vec p' \sim \Gamma }( \vec p'_J = \vec p_J )$ in polynomial time in $\enc{ \Gamma }$.
\end{enumerate}

For example, the first property holds for distributions encoded by Bayesian networks. The second one holds for structurally restricted classes of Bayesian networks, like polytrees \cite{koller2009probabilistic}.
The following result states that given these assumptions, the data complexity of the $\SHAP$ problem remains in polynomial time even for parameter distributions with correlations. 
This is shown exactly as in the proof of \autoref{pro:datacomplexity}.
\begin{prop}
    Let $\simi$ be a tractable similarity function and let $\pquery{Q}{\vec R}{\vec P}$ be a \emph{fixed} parameterized query such that $\pqueryinst{Q}{\vec p}(D)$ can be computed in polynomial time in $\enc{D} + \enc{\vec p}$ for all $\vec p \in \Tup[\vec P]$. Moreover, let $\PR$ be a class of distributions as described above. Then $\SHAP( \set{Q}, \PR, \simi )$ can be solved in polynomial time; that is, for every class $\Q$ of such queries, $\SHAP(\Q,\PR,\simi)$ can be solved in polynomial time in data complexity.
\end{prop}

Next, we explain how $\SHAP$ can be approximated via sampling. Consider an input $(Q,\vec p^*,D,\Gamma)$ of $\SHAP(\Q,\PR,\simi)$, where $Q = \pquery{Q}{\vec R}{\vec P}$ and $\len{ \vec p^* } = \ell$.
We rewrite $\Shap(i)$ as an expectation in a single probability space, instead of nested expectations in different spaces.
Let $\set{i}^0 = \emptyset$ and $\set{i}^1 = \set{i}$. Consider the following two-step random process, for $b \in \set{0,1}$: 
\begin{enumerate}
    \item Draw $J\subseteq [\ell]\setminus \set{i}$ according to $\Pi_i$.
    \item Draw $\vec p$ according to $\Gamma$, conditioned on having $\vec p$ agree with $\vec p^*$ on $J \cup \set{i}^b$. 
\end{enumerate} 

This process defines two joint probability distributions on pairs $(\vec p, J)$, one for $b=0$, and one for $b=1$. By $\Gamma^{i,b}$, we denote the corresponding marginal distribution over parameter tuples $\vec p$. 

\begin{prop}\label{pro:Shap-as-diff}
    We have
    \begin{equation}
        \label{eq:shapgamma-ib}
        \Shap(i) =
        \E_{ \vec p \sim \Gamma^{i,1} }\bigl[ \simi( \vec p, \vec p^* ) \bigr] -
        \E_{ \vec p \sim \Gamma^{i,0} }\bigl[ \simi( \vec p, \vec p^* ) \bigr]\text.
    \end{equation}
\end{prop}

\begin{proof}
Recall the two-step random process described at the beginning of \autoref{sec:cor-approx}. We have
\begin{equation}
    \Gamma^{i,b}(\vec p) =
    \sum_{ J \subseteq [\ell] \setminus \set{i} } 
    \Pi_i( J ) \cdot \Pr_{\vec p' \sim \Gamma}\bigl( \vec p' = \vec p \bigm\vert \vec p'_{ J \cup \set{i}^b } = \vec p^*_{ J \cup \set{i}^b } \bigr)\label{eq:defgammaib}
\end{equation}

Using $\Gamma^{i,b}$, we can rewrite the terms of \eqref{eq:shapleyv2exp} (after using linearity) as follows:
\begin{align*}
    \E_{ J \sim \Pi_i }[ \nu(J) ] &=
    \sum_{ J \subseteq [\ell] \setminus \set{i} } \Pi_i( J ) \cdot \E_{ \vec p \sim \Gamma }[ \simi( \vec p, \vec p^* ) \mid \vec p_J = \vec p^*_J ] \\ &=
    \sum_{ \vec p } \simi( \vec p, \vec p^* ) 
        \sum_{ J \subseteq [\ell] \setminus \set{i} } \Pi_i(J) \Pr_{\vec p' \sim \Gamma} ( \vec p' = \vec p \mid \vec p'_J = \vec p^*_J )\\&=
    \sum_{ \vec p } \simi( \vec p, \vec p^* ) \cdot \Gamma^{i,0}( \vec p ) =
    \E_{ \vec p \sim \Gamma^{i,0} }[ \simi(\vec p,\vec p^*) ]\text,
\end{align*}
and, analogously,
\[
    \E_{ J \sim \Pi_i }[ \nu( J \cup \set{i} ) ] =
    \E_{ \vec p \sim \Gamma^{i,1} }[ \simi( \vec p, \vec p^* ) ]\text.\qedhere
\]
\end{proof}

We say that $\PR$ \emph{admits efficient conditional sampling} if there exists a polynomial $\phi$ such that for all $\Gamma \in \PR$, all $\vec p^* \in \supp(\Gamma)$, and all $J \subseteq [\ell]$, the conditional distribution of $\Gamma$ subject to $\vec p_J = \vec p^*_J$ can be sampled in time $\phi(\enc{\Gamma} + \enc{\vec p^*})$.
This is again the case, for example, for structurally restricted classes of Bayesian networks.
By the structure of the two-step process defining $\Gamma^{i,b}$, we can efficiently sample from $\Gamma^{i,b}$ if we can efficiently sample from conditional distributions of $\Gamma$:

\begin{prop}\label{pro:conditional-sampling}
    If all conditional distributions of $\Gamma$ subject to conditions $\vec p_J = \vec p^*_J$ can be sampled in time polynomial in $\enc{ \Gamma }, \enc{\vec p^*}$, then the distributions $\Gamma^{i,b}$ can be sampled in time polynomial in $\enc{ \Gamma }, \enc{ \vec p^* }$.
\end{prop}

\begin{proof}
    The distribution $\Gamma^{i,b}$ was introduced as a two-step sampling process, first sampling $J$ from $\Pi_i$, and then sampling a parameter tuple from $\Gamma$ conditional on it agreeing with $\vec p^*$ on $J \cup \set{i}^b$. These two steps yield joint distributions on $2^{[\ell]\setminus\set{i}} \times \Tup[\vec P]$ for $b \in \set{0,1}$, and $\Gamma^{i,b}$ is the corresponding marginal distribution on $\Tup[\vec P]$. Hence, $\Gamma^{i,b}$ can be sampled by the two steps mentioned in the beginning, and returning only the sampled $\vec p$ from the second step.

    The distribution $\Pi_i$ can also be described by a two-step sampling process: 
    First, sample $k \in \set{0,1,\ldots, \ell - 1}$ uniformly at random. 
    Next, sample a subset $J$ of cardinality exactly $k$ from $[\ell] \setminus \set{i}$, uniformly at random. 
    For the second step, to sample uniformly from the $k$-element subsets of $[\ell] \setminus \set{i}$, we can proceed as follows: Initially, let $J = \emptyset$ and $\overline{J} = [\ell]\setminus \set{i}$. 
    Then, repeat the following $k$ times:
    \begin{enumerate}
        \item Let $j$ be a uniform sample from $\overline J$ (the uniform probability is $\frac{1}{\card{\overline{J}}}$).
        \item Remove $j$ from $\overline J$ and add it to $J$.
    \end{enumerate}
    Finally, return $J$.
    
    To see that this yields the correct distribution, consider any fixed set $K \subseteq [\ell] \setminus \set{i}$ of size $k$. 
    By summing over all permutations $j_1,\dots,j_k$ of $K$, we get
    \begin{align*}
        \Pr( J = K ) &= \sum_{ j_1,\dots,j_k } \Pr( j_1 ) \Pr( j_2 \mid j_1 ) \Pr( j_3 \mid j_1,j_2 ) \dotsm \Pr( j_k \mid j_1, \dots,j_{k-1} )\\
            &= \sum_{ j_1,\dots,j_k } \frac{1}{\ell-1} \frac{1}{ \ell-2 } \dotsm \frac{1}{\ell-1-(k-1)} 
            = \frac{k!(\ell-1-k)!}{(\ell-1)!} = \frac{1}{\binom{\ell-1}{k}}\text.
    \end{align*}

    Going back to the distribution $\Gamma^{i,b}$, we have just seen that the first step of the sampling process is possible in time polynomial in $\ell$ (and independent of any other inputs). 
    Together with the precondition of the statement, it follows that sampling from $\Gamma^{i,b}$ is possible in polynomial time.
\end{proof}

A similarity function $\simi$ is \emph{bounded} for a class of parameterized queries $\Q$ if there are $a \leq b$ such that for all $(Q,\vec p^1,\vec p^2)$, $Q \in \Q$, we have $a \leq \simi( \vec p^1, \vec p^2 ) \leq b$. For example, similarity measures, like $\Jaccard$, are usually $[0,1]$-valued. The following theorem states that, under the assumptions of efficient conditional sampling and boundedness, we have an additive FPRAS for the $\Shap$ score of a parameter.

\begin{thm}\label{thm:approx-shap}
    Let $\Q$ be a class of tractable parameterized queries, let $\simi$ a tractable similarity function such that the value range of $\simi$ is bounded for $\Q$, and let $\PR$ be a class of parameter distributions that admits efficient conditional sampling.
    Then, for all inputs $(Q,\vec p^*,D,\Gamma)$ and all $i\in[\ell]$, we can compute a value $S$ satisfying $\Pr( \abs{ S - \Shap(i) } < \epsilon ) \geq 1 - \delta$ in time polynomial in $\frac{1}{\epsilon}$, $\log\frac{1}{\delta}$, and the size of the input, that is, in combined complexity.
\end{thm}

\begin{proof}
    Let $N$ be some large enough number, the value of which we fix later. We draw a total of $2N$ independent samples: $N$ samples $\vec p_1^{(1)},\dots,\vec p_1^{(N)}$ from $\Gamma^{i,1}$, and $N$ samples $\vec p_0^{(1)}, \dots, \vec p_0^{(N)}$ from $\Gamma^{i,0}$. 
    By \cref{pro:conditional-sampling}, these samples can be obtained efficiently.
    Let $\widehat \simi_1$ and $\widehat \simi_0$ be the two sample means
    \[
        \widehat \simi_1 = \frac{1}{N} \sum_{ k = 1 }^{ N } \simi( \vec p_1^{(k)}, \vec p^* )
        \quad\text{and}\quad
        \widehat \simi_0 = \frac{1}{N} \sum_{ k = 1 }^{ N } \simi( \vec p_0^{(k)}, \vec p^* )
        \text.
    \]
    Then $\E\bigl[ \widehat \simi_1 \bigr] = \E_{ \vec p \sim \Gamma^{i,1} }\bigl[ \simi( \vec p, \vec p^* ) \bigr]$ and, likewise,
    $\E\bigl[ \widehat \simi_0 \bigr] = \E_{ \vec p \sim \Gamma^{i,0} }\bigl[ \simi( \vec p, \vec p^* ) \bigr]$, where the unannotated expectations relate to the sampling procedure. Hence, $\E\bigl[ \widehat \simi_1 \bigr] - \E\bigl[ \widehat \simi_0 \bigr] = \Shap(i)$ via \eqref{eq:shapgamma-ib}, and we therefore return the difference $S = \widehat \simi_1 - \widehat \simi_0$ as an approximation to $\Shap(i)$.
    
    Let $a \leq b$ be the constant bounds on $\simi$. By a corollary of Hoeffding's inequality (see \cite[Inequality (2.7)]{hoeffding1963probability}), for all $\epsilon > 0$ we have
    \[
        \Pr\Bigl( \abs[\big]{ S - \Shap(i) } < \epsilon \Bigr) \geq 
        1 - 2 \cdot \exp\biggl( - \frac{ N\epsilon^2 }{ (b-a)^2 } \biggr)\text.
    \]
    Thus, in order for our approximation to be in the interval $(\Shap(i) - \epsilon, \Shap(i) + \epsilon )$ with probability at least $1 - \delta$, it is sufficient to have
    \[
        \delta \geq 2 \cdot \exp\biggl( - \frac{N\epsilon^2}{(b-a)^2} \biggr) 
        \Leftrightarrow \frac{(b-a)^2}{\epsilon^2} \log\frac{2}{\delta} \leq N\text.
    \]
    That is, it is sufficient to use $\lceil \frac{(b-a)^2}{\epsilon^2} \log\frac{2}{\delta} \rceil$ samples from each distribution.
\end{proof} 
\section{Conclusions}\label{sec:conclusions}

We proposed a framework for measuring the responsibility of parameters to the result of a query using the $\Shap$ score. We studied the computational problem of calculating the $\Shap$ score of a given parameter value. We gave general complexity lower and upper bounds, and presented a complexity analysis for the restricted case of conjunctive queries and independent parameters. We also studied the extension of the study to conjunctive queries with filters, and showed an application of this extension to the analysis of why-not questions with the objective of quantifying the contribution of filters to the elimination of a non-answer. Finally, we discussed the complexity of approximate calculation and correlated parameters. 

The rich framework we introduced here offers many opportunities for future research. 
Especially important is the direction of aggregate queries, where the similarity between results accounts for the numerical values such as sum, average, median, and so on. For such queries, it is important to study numerical parameter distributions, which are typically continuous probability measures. It is also important to identify general tractability conditions for similarity measures and parameter distributions in order to generalize the upper bounds beyond the special cases that we covered here. Finally, it remains to evaluate empirically whether the potential uses discussed in \Cref{sec:shap} are effective in practice.

\section*{Acknowledgment}

The work of Martin Grohe, Benny Kimelfeld, and Christoph Standke was supported by the German Research Foundation (DFG) grants GR 1492/16-1 and KI 2348/1-1 (DIP Program). The work of Martin Grohe and Christoph Standke was supported by the German Research Foundation (DFG) grant GRK 2236/2 (UnRAVeL). The work of Amir Gilad was supported by the Israel Science Foundation (ISF) under grant 1702/24 and the Alon Scholarship.

\bibliographystyle{alphaurl}
\newcommand{\etalchar}[1]{$^{#1}$}

\end{document}